\documentclass[11pt,a4paper]{article}
\usepackage[english]{babel}
\usepackage[T1]{fontenc}
\usepackage{amsmath, amssymb, amsfonts, amsthm}
\usepackage{indentfirst}
\usepackage{bbm,enumitem}
\usepackage{mathrsfs}
\usepackage{float}
\usepackage{color}
\usepackage{units}
\usepackage{graphicx}
\usepackage{caption}
\usepackage{subfig}
\usepackage{setspace}
\usepackage{esint}
\usepackage{multirow}
\usepackage[authoryear]{natbib}
\usepackage{hyperref}
\usepackage{url}
\usepackage[left=2.25cm,right=2.25cm,top=2cm,bottom=3cm]{geometry}
\usepackage{tikz}
\usepackage{subfig}
\usetikzlibrary{decorations.pathreplacing}

\usepackage{chngcntr}
\usepackage{apptools}
\usepackage{titlesec}
\AtAppendix{
	\counterwithin{counter}{section}
	\titlelabel{Appendix \thesection: }
}

\newcounter{counter}
\newtheorem*{assumption*}{Assumption}
\newtheorem{proposition}[counter]{Proposition}
\newtheorem*{definition*}{Definition}

\newtheorem{lemma}[counter]{Lemma}

\newtheorem{corollary}[counter]{Corollary}
\newtheorem*{theorem*}{Theorem}
\newtheorem{theorem}{Theorem}

\newif\ifplotvhat
\newif\ifplottruth
\newif\ifplotcommit
\newif\ifplotcheap 
\def\valuefig{
	\draw[->] (0,0) -- (0,1.05);
	\draw[->] (0,0) -- (1.1,0) node[below] {$b_1$};
	\draw (1,0) node[below] {$1$};
	\draw (0,0) node[below left] {$0$};
	
	\def\pbar{0.5};
	\draw[dotted, opacity=0.5] (1,0) -- (1,1);
	\draw[dotted, opacity=0.5] (\pbar,0.125) -- (\pbar,0);
	\draw (\pbar,0) node[below]{$\bar{p}$};
	
	\ifplotvhat
	\draw[dashed, opacity=0.8, domain=0.1:1, smooth, variable=\x] plot ({\x}, {\x*\x*\x});
	\fi
	
	\ifplottruth
	\draw[truth, domain=0.5:1, smooth, variable=\x] plot ({\x}, {\x*\x*\x-\offsettruth});
	\draw[truth] (0,\offsettruth) -- (\pbar, \offsettruth);
	\draw[truth, dotted] (\pbar, 0.125-\offsettruth) -- (\pbar, \offsettruth);
	\fi
	
	\ifplotcommit
	\def\commitL{0.38};
	\def\commitR{0.62};
	\draw[commit, domain=\commitL:\commitR, variable=\x] plot ({\x}, {3/4*\x-1/4-\offsetcommit}); 
	\draw[commit, domain=\commitR:1, smooth, variable=\x] plot ({\x}, {\x*\x*\x-\offsetcommit});
	\draw[commit] (0,\offsetcommit) -- (\commitL, \offsetcommit);
	\draw[commit, dotted] (\commitL, {3/4*\commitL-1/4-\offsetcommit}) -- (\commitL, \offsetcommit);
	\draw[commit, dotted] (\commitR, {3/4*\commitR-1/4-\offsetcommit}) -- (\commitR, \commitR*\commitR*\commitR-\offsetcommit);
	\fi
	
	\ifplotcheap
	\def\cheapL{1/3+8/3*\offsetcheap};
	\def\cheapR{0.68};
	\draw[cheap, domain=\cheapL:\cheapR, variable=\x] plot ({\x}, {3/4*\x-1/4-\offsetcheap}); 
	\draw[cheap] (0, \offsetcheap) -- (\cheapL, \offsetcheap); 
	\draw[cheap, domain=\cheapR:1, variable=\x] plot ({\x}, {25/12*\x-125/108-\offsetcheap}); 
	\fi
}

\begin{document}
\title{Designing Social Learning\footnote{An earlier version of this paper was submitted to the University of Z{\"u}rich as a chapter of Smirnov's Ph.D. thesis (\citet{Smir20}). We are grateful to three anonymous referees, Mikhail Drugov, Renato Gomes, Johannes M{\"u}nster, Paolo Piacquadio, Marek Pycia, Armin Schmutzler, Jakub Steiner, Francesco Squintani, Bauke Visser, Nikhil Vellodi, Jidong Zhou, seminar participants at the University of Z{\"u}rich, J{\"o}nk{\"o}ping University, the University of M{\"u}nster, and participants of ESWEM 2020 and MaCCI 2023 conferences for valuable feedback and helpful comments.}
}
	
	\author{Aleksei Smirnov\footnote{Faculty of Economic Sciences, Higher School of Economics, Pokrovsky Boulevard 11, 109028, Moscow, Russia; e-mail: \href{mailto:assmirnov@hse.ru}{assmirnov@hse.ru}.}, Egor Starkov\footnote{Department of Economics, University of Copenhagen, {\O}ster Farimagsgade 5, 1353 K{\o}benhavn K, Denmark; e-mail: \href{mailto:egor@starkov.email}{egor@starkov.email}.}}
	
	\maketitle
	
	\begin{abstract}
		This paper studies strategic communication in the context of social learning. Product reviews are used by consumers to learn product quality, but in order to write a review, a consumer must be convinced to purchase the item first. When reviewers care about the welfare of future consumers, this leads to a conflict: a reviewer today wants the future consumers to purchase the item even when this comes at a loss to them, so that more information is revealed for the consumers that come after. We show that due to this conflict, communication via reviews is inevitably noisy, regardless of whether reviewers can commit to a communication strategy or have to resort to cheap talk.
		The optimal communication mechanism involves truthful communication of extreme experiences and pools the moderate experiences together.
		
		
	\end{abstract}
	
	\newpage

	\section{Introduction} \label{sec:INT}
	
	Whenever information is dispersed in society, the question of social learning arises: can society aggregate this information in a way that yields an efficient outcome for its members?
	In recent years, online customer reviews have emerged as a powerful tool of social learning: in multiple surveys of internet users, at least half of the respondents report using ratings and online reviews ``always'' or ``often'' to inform their purchasing decisions, and most respondents find reviews to be at least ``mostly reliable'' (\citet{CMA,Mintel,eMarketer}).
	Curiously, in one of these surveys, only about 10\% of respondents indicated that they find product reviews ``very reliable'' (\citet{eMarketer}). This skepticism can arise due to a variety of reasons, including the many ways in which sellers can tamper with reviews, such as through censorship or fake reviews.\footnote{See, e.g., \citet{LZ} for an exploration of the effects of fake reviews and \citet{SScens} for a model of censorship in product reviews.}
	However, in this paper we show that reviews can -- and should -- be noisy even in the absence of any intervention from sellers, with the noise stemming from \emph{how} customers write reviews.
	
	To understand the source of this noise, one must first ask \emph{why} customers write reviews. According to surveys, one of the most common responses to this question is ``to help other consumers'' (\citet{Trustpilot}).
	The desire to help other consumers make the right choice is often a sufficient incentive for people to spend time and effort writing a review. However, these altruistic concerns only seem to appear ex post, after the consumer has purchased and consumed the product, rather than ex ante. It is reasonable to expect that when deciding whether to buy a product and which one to buy, consumers focus primarily on their own expected utility from consumption, rather than the desire to provide helpful information to others.
	
	As we show, this time inconsistency in altruism leads to noise in product reviews. When product quality is uncertain, purchases have an informational externality, since they generate informative reviews that enable future consumers to make better decisions. This externality is not accounted for by consumers when they are deciding whether to purchase the product. It is, however, internalized by an altruistic reviewer, who may then want to mislead a future consumer into buying a product when it is socially, but not individually, optimal to do so.\footnote{The point that product evaluations produce a positive externality and are hence socially underprovided was made, among others, by \citet*{ARZ}, who proposed cash payments to alleviate this inefficiency.}
	
	We formalize the preceding argument in a model of product reviews. In this model, a sequence of consumers decides whether to buy a product of some uncertain quality and, if they do, what kind of review to write about their experience. A consumer only buys the product if she expects to derive high enough utility from consuming it. The realized consumption utility is informative about the product's quality and the expected utilities that other consumers would derive from consuming this product.\footnote{Unlike many other papers on experimentation, we do \emph{not} restrict attention to a model of exponential/Poisson bandits with binary utility outcomes, allowing instead for rich heterogeneity in utility realizations.} 
	The consumer can leave a review describing her consumption experience, and in doing so she wishes to maximize the welfare of future consumers considering the same purchase.
	Note, in particular, that we place the choice of how to communicate the consumption experiences directly in the consumers' hands. This approach differs from existing literature on experimentation in social learning, in which a centralized recommender platform makes these decisions (\citet*{KMP,CH}). It also differs from the extensive literature on observational learning, which takes the observation/communication technology for granted (\citet*{BHTW}). This difference and other distinctions are discussed in more detail in Section \ref{sec:LIT}.
	
	The myopic behavior at the purchasing stage and the altruistic desire to induce some experimentation with the product at the reviewing stage conflict with each other. We show that this conflict creates noise in communication through reviews. Instead of reporting their experiences truthfully, consumers may obfuscate their reviews to foster experimentation. This is true regardless of whether consumers can commit to some communication strategy ex ante (which can be interpreted as a shared social norm in reviewing) or select their review ex post, as in cheap talk models.
	
	In particular, we show that in the commitment scenario, the optimal communication strategy features perfect communication of extreme experiences (very high or very low utilities), but remains vague about experiences that would have put the next consumer close to indifference. The intuition is that the consumer would sometimes like to exaggerate product quality to induce a socially optimal purchase that the next consumer would not make if she were fully informed. However, for such a recommendation to be credible, it must also sometimes be issued when the product experience was, in fact, good -- meaning that such a review has to be vague, only imperfectly revealing the consumer's experience with the product.
	
	If a consumer cannot commit to a communication strategy, then despite the conflict of interest between her and a future consumer occurring only in a narrow set of circumstances -- when buying the product is socially optimal, but not individually optimal for the future consumer, -- the effects of this conflict propagate and distort communication in other cases as well. Communication is then always noisy in equilibrium, taking the interval structure often seen in cheap talk models when similar experiences are pooled in the same message. 
	
	Our paper provides a possible explanation for inflation in product reviews---namely, that reviews can sometimes be inflated in equilibrium, in order to deceive consumers into purchasing a product they would not have bought otherwise. This complements other possible explanations, including positive ratings being sponsored or faked by the sellers. Contrary to those explanations, in our case inflation arises endogenously as a result of interaction between consumers, with no intervention from the firm.
	Further, noise in communication is welfare-enhancing in our model, helping not only the seller, but also the consumers in aggregate.
	
	Our findings can be translated to settings beyond product reviews. In particular, they could be readily applied to scientific progress in the broad sense, contributing to the debate on whether negative results should be published.\footnote{See \citet{BPB20} and \citet{BMWetal} for evidence of publication bias in Economics, and \citet{NK20} and \cite{Bik} for examples of broader discussion on the topic.} 
	Our paper presents an argument \emph{against} publishing some negative results. Specifically, consider a world in which a series of researchers must decide whether to attempt to contribute to a broad research question (e.g., ``what is dark matter?'', ``how to detect a graviton?'', or ``are microplastics associated with negative health outcomes?''). Each researcher only attempts a project on the topic if they expect some positive findings to emerge (since there is no sense in looking for something that does not exist). In such a world, our findings suggest that mildly negative results should be suppressed if they discourage research on a topic that is socially beneficial. The idea is that such weak results may render the topic sufficiently unappealing for individual researchers to be unwilling to investigate it further, yet they are too weak to outweigh the potential social benefit if this research avenue were to be correct after all. Our model demonstrates that the benefits accruing from continuing exploration (by suppressing such weak negative findings) outweigh the losses from suppressing information.\footnote{\citet{CS24} also discuss dynamic externalities in research production, but explore a different question of the optimal \emph{novelty} of a research project.}
	
	Similarly, our results can be applied to designing informational nudges for present-biased individuals, as explored by \citet*{MSSW}. For example, consider a present-biased individual who is trying to develop a new positive habit (running in the mornings or cutting sugars/alcohol/tobacco from their diet), but is uncertain whether they can make this habit stick. We find that such an individual would benefit from keeping a journal of their experiences with the habit, where they would describe all extreme experiences accurately, positive and negative, but withhold mildly negative experiences. Recording the latter would run the risk of disincentivizing the present-biased self in the future from adhering to the habit, but would not constitute a sufficient reason for the forward-looking self to abandon the habit (unlike more extreme negative experiences).

	The remainder of the paper proceeds as follows. We review the relevant literature in Section \ref{sec:LIT}.
	We then formulate the model in Section \ref{sec:MOD}. Sections \ref{sec:EX} and \ref{sec:EQ} analyze a version of the model in which consumers are assumed to be able to commit to some communication structure, with Section \ref{sec:EX} exploring an illustrative three-period example, and Section \ref{sec:EQ} generalizing the insights to an infinite-horizon problem.
	Section \ref{sec:CT} then analyzes a model without commitment, where communication is cheap talk, looking at both the three-period example, and an infinite-horizon model.
	Section \ref{sec:CON} concludes. Proofs of all results are relegated to the Appendix.

	\section{Literature Review} \label{sec:LIT}
	
	Our paper belongs to the literature on social experimentation and the design of social learning; notable references include \citet*{KMP}, \citet*{MSS}, \citet{CH}, and \citet{CM}.
	This literature investigates the planner's problem of devising a communication protocol that incentivizes short-lived and/or selfish agents to experiment with a novel alternative more than they would on their own, for the sake of social benefit.\footnote{A separate literature explores optimal experimentation by groups of \emph{long-lived} agents, which explores the issues of free-riding that are fundamentally similar to the core friction in our model: see \citet{BH} and \citet*{KRC} for settings with full observability; \citet*{HRS} explore strategic communication under private payoffs. These papers focus on the repeated game effects -- i.e., devising the optimal reward and punishment strategies -- which are absent from games with short-lived agents, such as ours.}
	We explore what is effectively a decentralized version of the models mentioned above: instead of a single benevolent planner issuing recommendations, we allow each agent to communicate in the way they deem optimal. While the agents are altruistic and so have the same objective as the planner, they lack the \emph{memory} of a central planner. In our model, once one agent withholds some of the information about their experience, this information is lost forever -- unlike with the planner, who can remember privately all information that is not revealed publicly. This assumption is novel relative to the aforementioned papers and introduces a fundamentally new trade-off, since any attempt by the sender to mislead one receiver now carries not only the benefit of more social experimentation, but also the cost of subsequent receivers having worse information at their disposal. So, while the planner could simply make binary recommendations to consumers (whether to continue or stop experimenting/purchasing), in our decentralized setting, a richer communication strategy is optimal.\footnote{This trade-off is quite similar to that arising in problems of public communication to a heterogeneous audience; for some examples see \citet{AC} and \citet{IP}.}
	
	Closely related is the literature on herding and cascades in observational learning. In these models, agents receive private signals \emph{before} making a choice, and only observe past agents' \emph{choices}, which leads to inefficiencies (see \citet*{BHTW} for a more detailed overview of the topic). 
	\citet{Wol} considers a setting in which agents observe only past agents' outcomes, rather than actions, while \citet*{CBH} and \cite*{LSB} allow for various combinations of the two. These papers conclude that imperfect communication may increase the probability of cascades, and argue that perfect communication could alleviate the inefficiency. We show that despite this, perfect communication is not welfare-maximizing.
	\citet*{GHH} show that if only one action (such as buying a product) is observable to other consumers, then a cascade may only arise on this action. In our setting, the converse is true: cascades may only arise on ``not buying the product'', the action that does \emph{not} produce information.
	\citet{AK} and \citet*{SST} consider sequential observational learning with other-regarding preferences, but assume that agents only observe past agents' actions, and not their private information or their outcomes. 
	
	Social learning with strategic information provision was explored by \citet{SV15}; in their model conflict arises from senders' career concerns, which are absent in our case.
	\citet{LM} consider a sequential information acquisition model, where the information that an agent acquires is also observed by future agents. They show that information acquisition driven by myopic incentives can lead to long-run inefficiencies. 
	We show that tension between a forward-looking sender and a myopic receiver leads to noise in information provision.
	
	One version of the model we consider, where the sender cannot commit to a communication strategy, belongs to the literature on dynamic cheap talk.
	\citet*{AAK} explore hierarchical cheap talk, in which all communication must go through a chain of biased intermediaries. In our setting, all agents are perfectly benevolent. \citet*{RSV} consider a game in which a single sender repeatedly reports the state that follows a Markov process. The repeated nature of interaction shifts focus towards the repeated game effects, which are absent in our model. Highly related is the work of \citet{Chi}, who looks at a committee of agents with fully aligned interests and shows that perfect communication between the committee members may be infeasible. Friction in that model is driven by: (1) the agents' uncertainty about others' informedness and (2) the limited capacity of the communication channel. Communication constraints, however, play a crucial role in the result -- perfect communication would be an equilibrium if agents were able to report not only their recommendation but also reveal the information this recommendation is based on. Our model, in contrast, shows that technological constraints are not necessary for communication to be imperfect in settings with highly aligned interests.
	Contemporary work by \citet{BV24} considers a version of our three-period model without commitment, but with hard information {\`a} la \citet{DYE}. They show that such a communication technology produces outcomes very similar to the commitment technology we explore. We show, however, that this structure continues to hold in an infinite-horizon setting -- where more informative communication could, in principle, be optimal -- and we obtain our results without imposing any logconcavity assumptions on signal distributions.
	
	Our question is motivated by consumers' altruistic motives for writing reviews. This is mainly based on real-world consumers' own accounts regarding why they write reviews (\citet{Trustpilot}). However, the economic literature has also argued for a long time that people's real-world behavior often exhibits regard for others (see surveys by \citet{FS2}, \citet{Kon}, and \citet{Mei}).
	While evidence exists that market environments decrease the morality of participants' behavior (\cite{FS13}), it is still natural to expect at least some altruism towards fellow consumers in such settings.
	In particular, \citet{MZ} and \citet*{PRSX} find experimental evidence of altruistic motives in games of observational learning specifically.

	\section{The Model} \label{sec:MOD}
	
	\subsection{Primitives}
	
	Time is assumed to proceed in discrete periods: $t \in \mathcal{T}$. We consider versions with three periods, $\mathcal{T} = \{1,2,3\}$, and infinite horizon, $\mathcal{T} = \mathbb{N}$. All agents share a common discount factor $\beta$.
	
	\paragraph{Seller.} There is a single long-lived seller, who offers for sale a single product that he has in infinite supply at zero cost. Product quality $\theta$, which represents the average consumption utility of the product, can be either \emph{low} or \emph{high}: $\theta \in \{L,H\}$, with $0 \leqslant L < H$. The price of the product is fixed at $c > 0$; to avoid triviality we assume that $L < c < H$.
	
	\paragraph{Consumers.} In every period $t$, a single short-lived consumer $C_t$ arrives at the market. 	
	Upon entering the market, the consumer can either purchase the good at cost $c$ and subsequently leave a review as described further, or leave the market forever. The latter option yields the reservation utility normalized to $0$. 
	In case of purchase, the consumer receives a random consumption utility $v$, distributed according to a quality-contingent c.d.f. $F_\theta$ with mean $\theta$ and a respective p.d.f. $f_\theta$.
	We assume that both distributions have full support on the same open interval $\mathcal{V} = (\underline{v}, \overline{v}) \subseteq \mathbb{R}$, and their respective densities are continuously differentiable and bounded.\footnote{The support may be infinite: $\underline{v} = -\infty$ and $\overline{v} = +\infty$ are both admissible values. The common support assumption implies that no realized utility rules out one of the possible quality levels $\theta$.}
	In addition, we assume that the monotone likelihood ratio property (MLRP) holds:
	\begin{assumption*}[MLRP]
		Ratio $L(v) := \ln \left[\frac{f_{H}(v)}{f_{L}(v)}\right]$ is a strictly increasing and continuous function of $v$ on $\mathcal{V}$. Moreover, $\lim\limits_{v \rightarrow \underline{v}} L(v) = -\infty$, and $\lim\limits_{v \rightarrow \overline{v}} L(v) = +\infty$.\footnote{In the terminology of \citet*{SS_2000}, the assumption implies beliefs are unbounded.}
	\end{assumption*}
	
	All consumers are assumed to be Bayesian risk-neutral agents with lexicographic preferences, with the first-order preference being the consumer's own consumption utility and the second-order preference being the future consumers' expected consumption utility, discounted with factor $\beta$. 
	
	The consumer does not observe the product quality $\theta$, so her purchasing decision is based on her belief $p := \mathbb{P}(\theta = H)$ about the product quality, given the information available to her. In particular, the consumer purchases the product if and only if her expected consumption utility weakly exceeds the cost of purchase (we assume that the consumer purchases the product when indifferent):
	\begin{equation*}
		\theta(p) := p \cdot H + (1-p) \cdot L \geqslant c \iff p \geqslant \bar{p},
	\end{equation*}
	where $\bar{p} := \frac{c - L}{H - L}$. This purchasing strategy will be taken as given in what follows.

	\paragraph{Reviews.} If the good was purchased, the consumer then sends a message (writes a review) $m \in \mathcal{M}$ to subsequent consumers, describing her experience with the product. 
	A consumer's communication strategy maps her observed review history (described further) and her experience $v_t$ to $\varDelta (\mathcal{M})$. The optimal communication strategy maximizes the expected discounted sum of consumption utilities of all future consumers. 
	We consider two main versions of the model throughout the paper: in the \emph{commitment} scenario, we assume that consumers choose their communication strategy simultaneously with their purchasing strategy.\footnote{This commitment assumption can either be taken at face value -- that a consumer can choose a reviewing strategy in advance and follow it through -- or it can be seen through the lens of the planner's problem of devising an optimal \emph{social norm} that the consumers would use when writing reviews.} 
	In the \emph{cheap talk} version of the game, the consumer selects message $m \in \mathcal{M}$ after observing $v_t$ and not before.

	\paragraph{Timing.} Within a given period $t \in \mathcal{T}$, the sequence of events is as follows:
	\begin{enumerate}
		\item Consumer $C_t$ arrives at the market and observes all past reviews $\left(m_1, m_2, \ldots, m_{t-1}\right)$ and forms belief $p_t$ about the quality of the product.
		\item $C_t$ decides whether to purchase the product at cost $c$ or not and, in the \emph{commitment} model only, publicly selects her communication strategy.\footnote{We assume for simplicity that in the commitment model, every consumer's communication strategy is observable to all future consumers. Without this assumption, multiple equilibria could arise in which a specific communication strategy could be supported by some specific on- and off-the-equilibrium-path beliefs. It is straightforward that the equilibrium we find Pareto-dominates all such self-reinforcing equilibria.}
		\item After a purchase she receives random consumption utility $v_t \sim F_\theta$ and updates her belief about the product's quality.
		\item After a purchase, $C_t$ leaves review $m_t$ about her experience, according to her communication strategy in the \emph{commitment} model, or freely in the \emph{cheap talk} model. The review is observable to all subsequent consumers $(C_{t+1}, C_{t+2}, ...)$. If $C_t$ has not purchased the product, she leaves no review: $m_t = \varnothing$.
	\end{enumerate}

	\subsection{Histories, State Variables, Strategies} \label{sub:hist}
 
	\emph{Review history} $R_t := \left(m_1, m_2, \ldots, m_{t-1}\right)$ is a tuple consisting of all messages sent by consumers before period $t$. It constitutes the public history at the beginning of period $t$.\footnote{Formally, the observability of communication strategies in the commitment model implies that they must be included in public history, in addition to the realized messages. To economize on notation, we follow the convention in the field and suppress the dependence of histories and other objects on communication strategies. It should be understood that a given message $m_t$ is interpreted in the context of the observed communication strategy in the commitment model (and in the context of the equilibrium communication strategy in the cheap talk model, as is standard).} 
	We denote the \emph{public belief} about the quality of the product as 
	\begin{equation*}
		p_t := \mathbb{P}(\theta = H \mid R_t).
	\end{equation*} 
	The prior belief $p_1 = \mathbb{P}(\theta = H \mid \varnothing)$ is exogenously fixed and commonly agreed upon. 
	The \emph{private posterior belief} of consumer $C_t$ in the event she purchased and consumed the product is given by
	\begin{equation} \label{eq:bel_con1}
		b_t = b(p_t,v_t) := \mathbb{P}(\theta = H \mid p_t, v_t)
		= \frac{p_t f_H(v_t)}{p_t f_H(v_t) + \left(1-p_t\right)f_L(v_t)},
	\end{equation}
	where \eqref{eq:bel_con1} follows from the Bayes' rule.
	
	The public belief $p_t$ contains all payoff-relevant information available to $C_t$ at the time she decides whether to purchase the product. The pair of beliefs $p_t$ and $b_t$ summarizes all payoff-relevant information available to $C_t$ when she decides which message to send to subsequent consumers.
	Therefore, in what follows, we look at Markov equilibria (defined in Section \ref{sub:eqdef}), where $p_t$ is the time-$t$ \emph{public state} and a sufficient statistic of the review history $R_t$, and the tuple $(p_t,b_t)$ is $C_t$'s \emph{private state} and a sufficient statistic of her private history $(R_t,v_t)$. 
	
	The consumers' purchasing decisions are myopic, and so the uniquely optimal strategy (up to indifference) is ``buy if and only if $p_t \geqslant \bar{p}$''. Therefore, from this point onward, we focus on the consumers' \emph{communication} strategies.
	A behavioral strategy of consumer $C_t$ is $\mu_t$, where $\mu_t(m | p_t,b_t)$ denotes the probability with which $C_t$ sends message $m \in \mathcal{M}$ conditional on private state $(p_t,b_t)$. 
	In the commitment model, consumer $C_t$ chooses $\mu_t(p_t,\cdot)$ given public state $p_t$, while in the cheap talk model, $C_t$ chooses $\mu_t(p_t,b_t)$ given private state $(p_t, b_t)$.
	
	Let $\mathcal{M}_t(p_t) := \left\{ m \in \mathcal{M} \mid \exists \; b_t:\, \mu_t(m \mid p_t,b_t) > 0 \right\}$ denote the set of messages that are sent according to $\mu_t$ at $p_t$.
	Whenever $p_t \geqslant \bar{p}$, the public belief $p_{t+1}$ induced by an on-path message $m \in \mathcal{M}_t(p_t)$ is given by the Bayes' rule:
	\begin{equation} \label{eq:bel_con2}
		p_{t+1} = q_t(m | p_t) := \frac{p_t \cdot \int\limits_0^1 \mu_t(m \mid p_t,b_t) dF_H (v_t)}{p_t \cdot \int\limits_0^1 \mu_t(m \mid p_t,b_t) dF_H (v_t) + (1-p_t) \cdot \int\limits_0^1 \mu_t(m \mid p_t,b_t) dF_L (v_t)}
	\end{equation}
	If, on the other hand, $p_t < \bar{p}$, then $C_t$ does not buy the product, and message $m = \varnothing$ is sent, so $p_{t+1} = p_t$. In this case, $C_{t+1}$ has exactly the same information at the time she makes her purchasing decision as $C_t$ and does not purchase the product either.
	We shall refer to $p_{t+1}=q_t(m|p_t)$ as the \emph{posterior} or the \emph{induced public belief}.
	
	Given some strategy profile, let $\mathcal{P}_t(p_t) := \{q_t(m | p_t) \mid m \in \mathcal{M}_t(p_t)\}$ denote the set of all public posteriors induced by consumer $C_t$.
	We partition this set into the ``experimentation set'' $\mathcal{E}_t(p_t) := \left\{ q \in \mathcal{P}_t(p_t) \mid q \geqslant \bar{p} \right\}$, which includes all public posteriors $q$ that convince $C_{t+1}$ to buy the product, and the ``stopping set'' $\mathcal{S}_t(p_t) := \left\{ q \in \mathcal{P}_t(p_t) \mid q < \bar{p} \right\}$, which contains all $q$ that deter her from the purchase.
	Note that if $p_t \geqslant \bar{p}$ then $\mathcal{E}_t(p_t) \ne \emptyset$, as $p_t$ is a martingale (from the consumers' point of view).
	Conversely, as argued above, if $p_t < \bar{p}$ then $\mathcal{P}_t(p_t) = \mathcal{S}_t(p_t) = \{p_t\}$.
	The latter implies that all $q \in \mathcal{S}_t(p_t)$ are equivalent in the sense of shutting the market down from period $t+1$ onward.

	\subsection{Maximization Problem and Equilibrium Definitions}
	\label{sub:eqdef}
	
	Given a strategy profile for all future consumers, when consumer $C_t$ sends message $m$ at private state $(p_t,b_t)$, her continuation value (the discounted sum of future consumers' utilities) from doing so is equal to
	\begin{equation} \label{eq:value}
		V_t(m \mid p_t,b_t) := \mathbb{E} \left[ \sum\limits_{s=t+1}^{+\infty} \beta^{s-t-1} \cdot \mathbb{I}\left(p_s \geqslant \bar{p}\right) \cdot (v_s-c) \: \bigg| \: m, p_t, b_t \right].
	\end{equation}
	Implicit in \eqref{eq:value} is the fact that $m$ together with $p_t$ determines $p_{t+1}$ in equilibrium. It also embeds the dependence between future $v_s$ and future $p_s$, stemming from the future consumers' strategies.
	
	In the commitment model, we are looking for the Markov Perfect Equilibria (MPE) of the game, defined as follows.
	\begin{definition*}
		A Markov Perfect Equilibrium with Commitment consists of a collection of consumer communication strategies $\left\{ \mu_t(m | p_t,b_t) \right\}_{t \geqslant 1}$ and belief updating rules $b(p_t,v_t)$ and $\left\{ q_t(p_t,m) \right\}_{t \geqslant 1}$ such that the following conditions hold:
		\begin{itemize}
			\item \textbf{Belief Consistency:} condition \eqref{eq:bel_con1} holds for all $(p_t,b_t)$, and \eqref{eq:bel_con2} holds for all $p_t$ and all $m \in \mathcal{M}_t(p_t)$;
			\item \textbf{Optimality:} for any given $p_t$, and continuation strategies $\left\{\mu_t (m | p_\tau, b_\tau),\, q_t(p_\tau,m) \right\}_{\tau > t}$, communication strategy $\mu_t (m \mid p_t, b_t)$ maximizes $\mathbb{E}_{b} \left[ V_t(m \mid p_t,b) \mid p_t \right]$.
		\end{itemize}
	\end{definition*}
	
	The definition above presents a Markov equilibrium: $C_t$'s communication strategy does not depend on the review history $R_t$, except through belief $p_t$.
	The belief consistency condition ensures that all consumers use Bayes' rule to update their belief whenever possible. Optimality requires that every consumer chooses a communication strategy so as to maximize her ex ante value, as given by the expectation of \eqref{eq:value} over $v_t$ or, equivalently, $b_t$.
	In particular, maximizing the ex ante value means that $C_t$ commits to a communication strategy before observing the realized consumption utility $v_t$ (but she can condition on public state $p_t$).
	
	In turn, for the cheap talk game, the equilibrium definition is as follows.
	\begin{definition*}
		A Markov Perfect Equilibrium with Cheap Talk consists of a collection of consumer communication strategies $\left\{ \mu_t(m | p_t,b_t) \right\}_{t \geqslant 1}$ and belief updating rules $b(p_t,v_t)$ and $\left\{ q_t(p_t,m) \right\}_{t \geqslant 1}$ such that the following conditions hold:
		\begin{itemize}
			\item \textbf{Belief Consistency:} condition \eqref{eq:bel_con1} holds for all $(p_t,b_t)$, and \eqref{eq:bel_con2} holds for all $p_t$ and $m \in \mathcal{M}_t(p_t)$;
			\item \textbf{Ex Post Optimality:} for any given $p_t$, and continuation strategies $\left\{\mu_t (m | p_\tau, b_\tau),\, q_t(p_\tau,m) \right\}_{\tau > t}$, if $\mu_t(m^* \mid p_t,b_t) > 0$ then $m^* \in \arg \max_{m} V_t(m \mid p_t,b_t)$.
		\end{itemize}
	\end{definition*}
	
	In particular, the (ex ante) Optimality condition is replaced with Ex Post Optimality: the message sent by $C_t$ must now maximize her continuation value given her private state. To pin down the off-path beliefs, we assume that $q_t(m | p_t) = \varepsilon$ for all $p_t$ and all $m \notin \mathcal{M}_t(p_t)$ for some $\varepsilon < \bar{p}$.\footnote{This belief is admissible in equilibrium, since $v_t$ has full support, and thus at every $p_t$ there exists $v$ deemed possible on the equilibrium path such that $b(p_t,v) = \varepsilon$.}
	
	Without loss of generality, in both cases, we assume that $\mathcal{M} = [0,1]$ and restrict attention to equilibria with \emph{direct communication}, where $C_t$'s review simply prescribes the belief that $C_{t+1}$ must have after reading this and all other reviews: i.e., $m_t = q_t(p_t,m_t)$. This assumption is made for illustrative simplicity, allowing us to ignore the distinction between message $m_t$ and belief $p_{t+1}$ that it induces. Further, we assume that the set of stopping messages $\mathcal{S}_t(p_t)$ always consists of a single representative element whenever it is nonempty; as argued above, this is also without loss.

	\section{Three-Period Model with Commitment} \label{sec:EX}

	This section demonstrates the main insights in a simple three-period commitment model: $\mathcal{T} = \{1,2,3\}$. We look at a \emph{non-stationary} MPE with commitment, where the three consumers' strategies can be different. Suppose further that there is no discounting, $\beta=1$, so $C_1$ treats welfare of both $C_2$ and $C_3$ equally. We solve the example by backward induction. Fix some prior belief $p_1 \geqslant \bar{p}$. For the sake of this example, assume that consumption utilities are $v_t \sim \text{i.i.d.}\ \mathcal{N} \left(\theta, \sigma^2\right)$. 
	
	In period $t=3$, $C_3$ purchases the product if and only if $\theta(p_3) = H \cdot p_3 + L \cdot (1-p_3) \geqslant \bar{p}$, and her messaging strategy is irrelevant, since no consumers arrive at the market after her. We hence continue straight to $t=2$.

	\subsection{Second Period} \label{sub:EX2}
	
	In the second period, if $p_2 < \bar{p}$ then, as mentioned in the model setup, the game effectively ends. $C_2$ does not buy the product, writes no review, so $p_3 = p_2 < \bar{p}$, and $C_3$ does not buy the product either. Payoffs of $C_2$ and $C_3$ are zero in this case.
	Conversely, if $p_2 \geqslant \bar{p}$ then $C_2$'s continuation value equals $C_3$'s expected consumption utility: $V_2(m | p_2, b_2) = \theta(q(m | p_2)) - c$. Therefore, at the review stage $C_2$ would prefer to act in the best interest of $C_3$. Truthful communication, where $C_2$ reports $m_2 = p_2$, is thus optimal.
	
	However, perfect communication is not necessary to achieve the maximal payoff for $C_3$.
	Note that the only piece of information relevant to $C_3$ is whether to buy the product or not. She cannot make use of more precise information to make better recommendations to future consumers because there are no future consumers. Therefore, $C_2$ can achieve her optimum via a simple binary communication strategy that sends two distinct messages depending on whether $b_2 \geqslant \bar{p}$ or $b_2 < \bar{p}$.

	\subsection{First Period} \label{sub:EX1p} \label{sub:EX1c}
	
	As shown above, perfect communication is optimal for $C_2$. We now show that the same is not true for $C_1$. 
	We begin this section by analyzing $C_1$'s continuation value $V_1(m_1 \mid p_1,b_1)$ as a function of the message she sends or, equivalently, of the public belief $p_2 = m_1 = q_1(m_1 | p_1)$ she induces. 
	
	Recall that we assumed in the example setup that $p_1 \geqslant \bar{p}$, otherwise $C_1$ does not buy the product and all values are zero. 
	Then $C_1$ buys the good and receives utility $v_1$. If she sends $m_1 \leqslant \bar{p}$, then neither $C_2$ nor $C_3$ will buy the product, and $C_1$'s continuation value is zero.
	If she sends $m_1 \geqslant \bar{p}$, then $C_2$ purchases the product, obtains utility $v_2$, and reveals her posterior $b_2$ truthfully. Following that, $C_3$ purchases the product if and only if $p_3 \geqslant \bar{p} \Leftrightarrow b_2 = b(p_2,v_2) \geqslant \bar{p} \Leftrightarrow v_2 \geqslant \bar{v}_2(p_2)$, where $\bar{v}_2(p_2)$ is defined as the solution to $b\left(p_2,\bar{v}_2(p_2) \right) = \bar{p}$.
	Therefore, $C_1$'s continuation value from inducing public belief $p_2$ is given by
	\begin{equation*}
	V_1 (p_2 \mid p_1, b_1) =
	\begin{cases}
		\hat{V}(p_2 \mid p_1, b_1) & \text{ if } m_1 \geqslant \bar{p},
		\\
		0 & \text{ if } m_1 < \bar{p},
	\end{cases}
	\end{equation*}
	where $\hat{V}(p_2 \mid p_1, b_1)$ is $C_1$'s continuation value if $C_2$ buys the product:
	\begin{align}
		\nonumber
		\hat{V}(p_2 \mid p_1, b_1) &:= \mathbb{E} \left[ (v_2-c) + (v_3-c) \cdot \mathbb{I}\{v_2 \geqslant \bar{v}_2(p_2) \} \mid b_1 \right]
		\\
		\label{eq:val_nonbab}
		&= \theta(b_1) - c + b_1 \cdot \left(1-F_H(\bar{v}_2(p_2))\right) \left(H - c\right) + \left(1-b_1\right) \cdot \left(1-F_{L}(\bar{v}_2(p_2))\right) \left(L - c\right).
	\end{align}

	\begin{figure}
		\centering
		\begin{tikzpicture}[xscale=6,yscale=3.2]
			\draw[->] (0,-0.8) -- (0,1.2) node[left] {$\hat{V}(p_2 \mid p_1, b_1)$};
			\draw[->] (0,0) -- (1.1,0) node[below] {$p_2$};
			\draw[dotted] (1,-0.8) -- (1,1.2);
			\draw (1,0) node[below left] {$1$};
			\draw (0,0) node[below left] {$0$};
			
			\def\pbar{0.33333};
			\draw[dotted] (\pbar,1.2) -- (\pbar,-0.8);
			\draw (\pbar,0) node[below right]{$\bar{p}$};
			\draw (\pbar,0.03) -- (\pbar,-0.03);
			
			\draw[cyan,dotted] (0.2,0) -- (0.2, -0.361);
			\filldraw[cyan] (0.2, -0.361) circle [x radius=0.01, y radius=0.02];
			\draw (0.2,-0.361) node[below]{\tiny $b_1=\frac{1}{5}$};
			\draw [cyan, dashed] plot [smooth] coordinates { (0,-0.4) (0.01,-0.3999) (0.02,-0.3999) (0.03,-0.3998) (0.05,-0.3986) (0.07,-0.3952) (0.1,-0.3863) (0.15, -0.369) (0.2, -0.361) (0.25, -0.3703) (0.3, -0.3977) (0.33, -0.4216)};
			\draw [cyan] plot [smooth] coordinates { (0.33, -0.4216) (0.35, -0.4401) (0.4, -0.4924) (0.45, -0.5492) (0.5, -0.6056) (0.55, -0.6578) (0.6, -0.7028) (0.65, -0.739) (0.7, -0.7658) (0.75, -0.7835) (0.8, -0.7937) (0.85, -0.7984) (0.9, -0.7998) (0.92, -0.7999) (0.95, -0.7999) (0.97, -0.7999) (0.99, -0.8) (0.999, -0.8) (1, -0.8)};
			
			\draw[blue,dotted] (0.5,0) -- (0.5, 1.0488);
			\filldraw[blue] (0.5, 1.0488) circle [x radius=0.01, y radius=0.02];
			\draw (0.5, 1.0488) node[above]{\tiny $b_1=\frac{1}{2}$};
			
			\draw [blue, dashed] plot [smooth] coordinates { (0, 0.5) (0.0001, 0.5) (0.001, 0.5) (0.005, 0.5) (0.01, 0.5) (0.02, 0.5) (0.03, 0.5004) (0.05, 0.5041) (0.07, 0.5149) (0.1, 0.5472) (0.15, 0.6359) (0.2, 0.743) (0.25, 0.8443) (0.3, 0.9269) (0.3333, 0.9692)};
			
			\draw [blue] plot [smooth] coordinates { (0.3333, 0.9692) (0.35, 0.9866) (0.4, 1.0242) (0.45, 1.0434) (0.5, 1.0488) (0.55, 1.045) (0.6, 1.0361) (0.65, 1.0255) (0.7, 1.0157) (0.75, 1.0081) (0.8, 1.0033) (0.85, 1.001) (0.9, 1.0001) (0.92, 1) (0.95, 1) (0.97, 1) (0.99, 1) (0.999, 1) (1,1)};
			
			\draw[red,dotted] (0.3333,0) -- (0.3333, 0.1949);
			\filldraw[red] (0.3333, 0.1949) circle [x radius=0.01, y radius=0.02];
			\draw (0.3333, 0.1949) node[above]{\tiny $b_1=\frac{1}{3}$};
			\draw [red, dashed] plot [smooth] coordinates { (0,0) (0.03, 0.0003) (0.05, 0.0026) (0.07, 0.0093) (0.1, 0.0286) (0.15, 0.0777) (0.2, 0.1297) (0.25, 0.1695) (0.3, 0.191) (0.3333, 0.1949) };
			\draw [red] plot [smooth] coordinates {	(0.3333, 0.1949) (0.35, 0.194) (0.4, 0.1817) (0.45, 0.1587) (0.5, 0.1297) (0.55, 0.099) (0.6, 0.07) (0.65, 0.0452) (0.7, 0.026) (0.75, 0.0128) (0.8, 0.005) (0.85, 0.0013) (0.9, 0.0002) (0.95, 0) (0.97, 0) (0.99, 0) (0.999, 0) (1,0)};
			
			\matrix [draw, fill=white, below left] at (current bounding box.north east) {
				\draw [blue] ++(-0.2,0) -- ++(0.4,0) node[black,right] {\footnotesize $b_1=1/2$}; \\
				\draw [red] ++(-0.2,0) -- ++(0.4,0) node[black,right] {\footnotesize $b_1=1/3$}; \\
				\draw [cyan] ++(-0.2,0) -- ++(0.4,0) node[black,right] {\footnotesize $b_1=1/5$}; \\
			};
		\end{tikzpicture}
		\caption{$\hat{V}(p_2 \mid p_1,b_1)$ as a function of $p_2$. \label{fig:example}}
		\small
		\emph{Note: The parameter values are $H=3, L=0, c=1$ (so $\bar{p}=1/3$), $\sigma = 4$. Since $\hat{V}(p_2 \mid p_1,b_1)$ only coincides with $V_1(p_2 \mid p_1,b_1)$ when $p_2 \geqslant \bar{p}$, we use dashed lines for values at $p_2 < \bar{p}$.}
	\end{figure}
	
	Analyzing \eqref{eq:val_nonbab}, we can identify several important properties of $\hat{V}(p_2 \mid p_1, b_1)$, which are plotted in Figure \ref{fig:example}.
	We start by noting that $\hat{V}(p_2 \mid p_1,b_1)$ is single-peaked in $p_2$, with a peak at $p_2 = b_1$. This means that \emph{conditional} on $C_2$ buying the product, $C_1$ wants to communicate truthfully and induce the correct belief, $p_2=b_1$. 
	To see this, observe that
	\begin{equation*}
		\frac{\partial \hat{V}(p_2 \mid p_1, b_1)}{\partial p_2} = \left(1-b_1\right) \cdot f_L(\bar{v}_2(p_2)) \cdot \frac{\bar{p} \cdot \sigma^2}{(1 - \bar{p}) \cdot p_2 (1-p_2)} \left( \frac{b_1}{1-b_1} \cdot \frac{1 - \bar{p}}{\bar{p}} \cdot \frac{f_H(\bar{v}_2(p_2))}{f_L(\bar{v}_2(p_2))} - 1\right),
	\end{equation*}
	where single-peakedness follows from $\frac{f_H(\bar{v}_2(p_2))}{f_L(\bar{v}_2(p_2))}$ being strictly decreasing in $p_2$ (due to MLRP), and the term multiplying the bracket being positive. 
	The peak must satisfy 
	\begin{align*}
		\frac{\partial \hat{V}(p_2 \mid p_1, b_1)}{\partial p_2} = 0
		&&\iff &&
		\frac{\bar{p}}{1-\bar{p}} = \frac{b_1}{1-b_1} \cdot \frac{f_{H}(\bar{v}_2(p_2))}{f_{L}(\bar{v}_2(p_2))},
	\end{align*}
	which, by \eqref{eq:bel_con1} and the definition of $\bar{v}_2(p_2)$, is equivalent to $p_2 = b_1$.
	We conclude that \emph{after} $C_1$ convinces $C_2$ to buy, she has no incentives to distort her review (and the same is true if $C_1$ convinces $C_2$ to pass). Therefore, the only incentives for $C_1$ to misreport her experience can come from the need to convince $C_2$ to buy the product.
	
	To see that such a motive is indeed present, observe that $\hat{V}(p_2 \mid p_1,b_1)$ is positive for $b_1 = \bar{p} - \varepsilon$ for at least some $\varepsilon > 0$. This means that $C_1$ sometimes wants $C_2$ to purchase a product that yields negative expected consumption utility. This is due to the social value of experimentation (i.e., of information generated by $C_2$'s purchase), which is internalized by $C_1$ in her communication strategy, but not by $C_2$ in her purchasing strategy. 
	To see this, note first that $\hat{V}(p_2 \mid p_1,\bar{p}) > 0$ for all $p_2$, since $F_H(\bar{v}_2(p_2)) < F_L(\bar{v}_2(p_2))$. Function $\hat{V}(p_2 \mid p_1, b_1)$ is continuous in $b_1$, hence it is also strictly positive in some neighborhood of $b_1 = \bar{p}$. This implies that $C_1$ strictly prefers to induce $p_2 \geqslant \bar{p}$ for at least some posteriors $b_1 < \bar{p}$, as compared to inducing $p_2 < \bar{p}$: she wants $C_2$ to purchase the product despite herself believing that this is not myopically optimal.  
	
	However, a recommendation to buy the product must be credible: if it is sometimes made after $b_1 < \bar{p}$, it must also be sometimes made after $b_1 > \bar{p}$ for the posterior $p_2$ to be $p_2 \geqslant \bar{p}$. 
	Suppose then that instead of being perfectly informative, $C_1$ sends the same message $\bar{p}$ after all $b_1$ in some $\varepsilon$-neighborhood of $\bar{p}$. Reporting $b_1 > \bar{p}$ inaccurately yields a welfare loss, but it is only of the order $\varepsilon^2$, since $\hat{V}(\bar{p} \mid p_1,b_1)$ is tangent to $\hat{V}(b_1 \mid p_1,b_1)$ at $b_1 = \bar{p}$. Intuitively, this little lie does not affect $C_2$'s expected payoff, but may lead $C_2$ to discourage $C_3$ from buying a product that actually has a positive expected payoff. Both the expected payoff that $C_3$ would be missing out on and the probability of such a mistake are roughly proportional to $\varepsilon$, hence the total loss is of order $\varepsilon^2$.
	On the other hand, misreporting $b_1 = \bar{p} - \varepsilon$ as $\bar{p}$ yields a benefit approximately equal to $\varepsilon \cdot \hat{V}(\bar{p} \mid p_1,\bar{p})$, due to inducing a purchase from $C_2$ in cases where $C_2$ would have passed if she knew $b_1$. As was previously established, $\hat{V}(\bar{p} \mid p_1,\bar{p}) > 0$, hence the benefit is of order $\varepsilon$, and it outweighs the cost for at least some small $\varepsilon$. Therefore, garbling information when $b_1$ is close to $\bar{p}$ is indeed beneficial for $C_1$.
	Further, belief $p_2$ induced by the pooling message must be exactly $\bar{p}$. If it was lower, $C_2$ would not buy the product, which is $C_1$'s goal. If it was higher, then $C_1$ could reduce the pooling interval by communicating the highest of pooled $b_1$ truthfully instead -- which $C_1$ prefers, by the single-peakedness argument above.\footnote{While such a deviation would also decrease $C_1$'s value from the pooling message, the convexity of $\hat{V}(b_1 \mid p_1,b_1)$ in $b_1$ together with Jensen's inequality suggests that the deviation is indeed beneficial in expectation.}
	A similar argument suggests that the optimal pooling interval must be convex.
	
	In the end, the following communication strategy for $C_1$ is optimal, according to the intuitive argument above: send message $\bar{p}$ when $b_1$ is in some neighborhood of $\bar{p}$, and report $b_1$ truthfully otherwise.
	Figure \ref{fig:optimal_value} plots the value attained by $C_1$ under this communication strategy. 
	In the following section, we demonstrate formally in the context of an infinite-horizon model that such a communication strategy is indeed optimal.\footnote{The argument for the three-period model can be obtained by following the same steps as in the proof of Theorem~\ref{thm:commitment} below. A contemporary paper by \citet{BV24} formally derives an optimal strategy in a similar three-period model and shows that it is indeed of the form that we present here.} 
	
	\begin{figure}
		\centering
		\subfloat[][Truthful communication.]{
		\begin{tikzpicture}[xscale=6.5,yscale=6, 
				vhat/.style={dashed},
				truth/.style={color=violet!70!white, line width=1},
				commit/.style={color=red!70!yellow, line width=1},
				cheap/.style={color=cyan, line width=1},
			]
			\plotvhattrue
			\plottruthtrue
			\plotcommitfalse
			\plotcheapfalse
			
			\def\offsettruth{0.005};
			\def\offsetcommit{0.01};
			\def\offsetcheap{0.006};
			
			\valuefig
			
			\def\stroke{++(-0.2,0) -- ++(0.4,0)};
			\matrix [draw, fill=white, below right] at (0.05,1.0) {
				\draw [vhat] \stroke node[black,right] {\footnotesize $\hat{V}(b_1 \mid p_1,b_1)$}; \\ 
				\draw [truth] \stroke node[black,right] {\footnotesize $V_1(b_1 \mid p_1,b_1)$}; \\ 
			};
		\end{tikzpicture}
		}
		\subfloat[][Optimal communication with commitment.]{
		\begin{tikzpicture}[xscale=6.5,yscale=6, 
				vhat/.style={dashed},
				truth/.style={color=violet!70!white, line width=1},
				commit/.style={color=red!70!yellow, line width=1},
				cheap/.style={color=cyan, line width=1},
			]
			\plotvhattrue
			\plottruthfalse
			\plotcommittrue
			\plotcheapfalse
			
			\def\offsettruth{0.003};
			\def\offsetcommit{0.005};
			\def\offsetcheap{0.006};
			
			\valuefig
			
			\def\stroke{++(-0.2,0) -- ++(0.4,0)};
			\matrix [draw, fill=white, below right] at (0.05,1.0) {
				\draw [vhat] \stroke node[black,right] {\footnotesize $\hat{V}(b_1 \mid p_1,b_1)$}; \\ 
				\draw [commit] \stroke node[black,right] {\footnotesize $V_1^{Comm}(m^* \mid p_1,b_1)$}; \\
			};
		\end{tikzpicture}
		}
		\caption{$C_1$'s continuation value given different communication strategies. \label{fig:optimal_value}}
		
		\footnotesize\emph{Note: Dashed line in both figures represents $C_1$'s continuation value from truthtelling conditional on $C_2$ buying the product. Solid purple line in panel (a) plots $C_1$'s value from truthful communication; solid red line in panel (b) plots $C_1$'s value from optimal communication with commitment (and $m^*$ denotes the message sent by the optimal communication strategy).}
	\end{figure}

	\section{Infinite-Horizon Model with Commitment} \label{sec:EQ}

	This section describes the equilibrium of the infinite-horizon commitment game with discounting: $\mathcal{T} = \mathbb{N}$; $\beta < 1$.
	Specifically, we are looking for \emph{stationary} MPE with Commitment (hereinafter simply ``equilibria'') of the game, where the consumers' communication strategies $\mu (m | p_t, b_t)$ and updating rules $q(p_t,m)$ do not depend on $t$.
	It is not ex ante clear whether the equilibrium with an infinite horizon would necessarily look the same as in the three-period example, since in the latter $C_1$ only needed to deceive one following consumer, $C_2$, all to benefit $C_3$. With the infinite horizon, $C_t$ may similarly want to induce more experimentation by deceiving $C_{t+1}$ while supplying her with enough information to give good recommendations to $C_{t+2}$ and onward. However, two new considerations arise. First, both the benefit and the cost of the noise in $C_t$'s review now affect an infinite number of future consumers, who benefit from $t+2$ onward from more experimentation at $t+1$, but suffer from $t+1$ onward from acting on worse information than they could have had. Second, $C_t$ now has to take into account that $C_{t+1}$ may want to deceive $C_{t+2}$ to benefit those from $t+3$ onward, and $C_{t+2}$ may want to deceive $C_{t+3}$, and so on. Would these considerations potentially lead $C_t$ to alter her communication strategy? 
	We show below that this is not the case, and the optimal communication strategy is similar to the one arising in the three-period example: pooling experiences close to $\bar{p}$ and perfectly informative otherwise. In particular, neither perfect communication is optimal, nor is it optimal for a consumer to introduce other kinds of noise to help the next consumer(s) trick those who follow.
	
	The statement of the main result below and the argument behind it mirror the conclusions from the three-period model (Section \ref{sub:EX1c}): $C_t$'s desire to inflate the review of a marginally-bad item for the sake of social experimentation results in garbled communication being optimal. The reviewer ends up issuing the same review for a product that she believes is barely good enough and for a product that is subpar but not bad enough for her to outright reject. After reading such a review, $C_{t+1}$ is exactly indifferent between buying the product and not. Conversely, if $C_t$ is sufficiently confident in her judgment of the product quality, then it is optimal for her to report her experience truthfully.
	
	\begin{theorem} \label{thm:commitment}
	There exists a unique equilibrium such that for any $p_t \geqslant \bar{p}$, $C_t$'s communication strategy is characterized by cutoffs $0 < l(p_t) < \bar{p} < r(p_t) < 1$ such that:
		\begin{enumerate}
			\item For all $b_t \in [0,l(p_t))$, $C_t$ truthfully reveals her private belief $b_t$: $m_t = b_t$, i.e., the experimentation stops.
			\item For any $b_t \in (r(p_t), 1]$, $C_t$ truthfully reveals her private belief $b_t$: $m_t = b_t$.
			\item For all $b_t \in \left[l(p_t), r(p_t)\right]$, $C_t$ sends message $m_t = \bar{p}$.
		\end{enumerate}
	Any other equilibrium is payoff-equivalent to this equilibrium.
	\end{theorem}
	
	The conclusion that some noise is optimal is driven by the lexicographic nature of the consumers' preferences. When writing a review, $C_t$ maximizes welfare, so she would like $C_{t+1}$ to sometimes buy the product when it is myopically suboptimal for her in order to generate more information about product quality. 
	Reviewer $C_t$ thus faces an incentive to sometimes send message $m \geqslant \bar{p}$ when $b_t < \bar{p}$. The benefit of such an upwards distortion for an interval of beliefs $b_t \in (\bar{p}-\varepsilon, \bar{p})$ is approximately equal to $\varepsilon \cdot V(\bar{p} \mid p_t, \bar{p}) \sim \mathcal{O}(\varepsilon)$. 
	For such a message to induce a posterior belief $q(m | p_t) \geqslant \bar{p}$, this distortion must be balanced by a proportionate downwards distortion: the same message $m$ must be sent after some $b_t > \bar{p}$. 
	Such a downwards distortion is costly, since it not only conceals some information about the state from all future consumers, but it also decreases the amount of experimentation from $t+2$ onward. However, the total losses from distorting the review downwards after $b_t \in (\bar{p}, \bar{p} + \varepsilon)$ are of order $\mathcal{O}(\varepsilon^2)$, because $V(p_{t+1} \mid p_t, b_t)$ is continuous and smooth in $p_{t+1} \geqslant \bar{p}$. 
	For small enough $\varepsilon$, the gains outweigh the costs, so it is optimal to send a vague review when $C_t$ is sufficiently uncertain of whether the product is worth buying, as opposed to revealing the private posterior $b_t$ exactly.
	
	Part 3 of the Theorem states that the pooling interval is convex, i.e., it is an interval of private posteriors $b_t$ around $\bar{p}$ that is pooled into the vague message, as opposed to disparate experiences from both sides. In particular, the incentive to lie only exists for $b_t$ \emph{slightly} below $\bar{p}$. If $b_t \ll \bar{p}$, then $C_t$'s opinion of the product is so low that the product is not worth experimenting with any further, even after accounting for the informational externality of future purchases, hence reporting such beliefs truthfully is optimal. 
	On the other hand, the cheapest way to make the ``buy'' recommendation credible is to pool it with experiences $b_t$ just above $\bar{p}$, as opposed to $b_t \gg \bar{p}$. In the latter case, $C_t$ would need to distort her experience heavily downwards relative to $b_t$, and the welfare cost of such a distortion increases \emph{faster} in $b_t$ than the persuasive power generated by pooling such very positive experiences into message $m_t=\bar{p}$. Credibility of message $m_t = \bar{p}$ is, therefore, best achieved by distorting the marginally positive experiences $b_t \in [\bar{p}, r(p_t))$, while the more positive experiences are best reported truthfully in order to allow future consumers to make precise recommendations.
	
	Interestingly, Theorem \ref{thm:commitment} can be related to empirical evidence suggesting that consumers are most likely to write reviews after either very positive, or very negative experiences (\citet{Trustpilot}).
	Our result suggests that this reviewing strategy is socially optimal, since leaving no review can be treated as a pooling (catch-all) message $m = \bar{p}$. Of course, in the real world, many consumers do not leave a review regardless of their experience with the product, but our model can be extended to account for that without affecting the main result.
	
	The formal proof of Theorem \ref{thm:commitment} in the Appendix proceeds in three main steps. First, we show that pooling is only beneficial around the cutoff. The second step shows that gains from pooling over an arbitrarily small interval of posteriors will be of the first order, while losses will be of the second order, meaning that some noise is always optimal in equilibrium. The final step shows that the equilibrium of the specified form (with pooling around $\bar{p}$ and truthtelling otherwise) does, indeed, exist, and is unique up to payoff equivalence.

	\section{Cheap Talk Model} \label{sec:CT}
	
	The two previous sections analyze a model with commitment and derive the optimal ``social norm'' -- the planner's welfare-maximizing guideline for altruistic consumers. A question inevitably arises: would consumers abide by this social norm? In this section, we address it by considering a cheap talk version of the model, where consumers write their review \emph{after} observing their utility realization $v_t$ (the equilibrium definition is adjusted accordingly, see Section \ref{sub:eqdef}.). We show that the optimal norm cannot be an equilibrium outcome in such a setting. This section begins by analyzing the three-period model, where we show that if communication is informative in both $t=1$ and $t=2$, then it is noisier than the optimal norm, leading to too much experimentation and lower welfare. We then proceed to analyze an infinite-horizon model, where we show that noise is robust: it \emph{must} arise in equilibrium regardless of how other consumers behave in the continuation equilibrium, so long as they produce meaningful information. While truthful communication is possible in a given period, this can only happen if social learning is expected to fully stop afterwards.

	\subsection{Three-Period Model} \label{sub:int3}
	
	In this section we analyze the three-period cheap talk model. As in Section \ref{sec:EX}, we include current period $t$ in both private and public states, allowing strategies to be time-dependent, and focus on non-stationary MPE with Cheap Talk. The first thing to note is that the analysis of periods $t=3$ and $t=2$ is completely analogous to what is presented in Section \ref{sec:EX} for the model with commitment. Specifically, in period 3, communication is irrelevant, since no one else arrives at the market to read $C_3$'s review. Then in period 2, $C_2$ has no reason to misreport her experience, and hence truthful reporting is an equilibrium.\footnote{Unlike in Section \ref{sec:EX}, in cheap talk model there exist continuation equilibria at $t=2$ that are \emph{payoff-distinct} from the perfect communication equilibrium. However, for the purpose of exposition, in this section we select continuation equilibria with perfect communication in period 2 in order to simplify the analysis of communication in period 1.}
	Furthermore, the analysis of $C_1$'s continuation payoffs in Section \ref{sub:EX1p} carries over to this setting as well.
	Therefore, we focus on analyzing the equilibrium first-period communication strategy and show that the equilibrium communication strategy must have an interval structure: i.e., there exists a partition $0 = \Delta_0 < \Delta_1 < \Delta_2 < \ldots = 1$ and messages $m_0, m_1, \ldots$ such that if $b_1 \in (\Delta_{j}, \Delta_{j+1})$ then $\mu(m_j \mid p_1, b_1) = 1$. 
	
	Suppose that $\mathcal{S}_1(p_1)$ is nonempty, i.e., there exists a review $m_0 \in \mathcal{S}_1(p_1)$ that will prevent $C_2$ from buying the product. Then this review will be used by $C_1$ if her private posterior $b_1$ is low enough. To see this, recall that $C_1$'s continuation value $V(p_2 \mid p_1, b_1)$ is defined in \eqref{eq:value}, and her continuation value conditional on $C_2$ buying the item, $\hat{V}(p_2 \mid p_1, b_1)$, is given by \eqref{eq:val_nonbab}. For $b_1$ close to $0$, expression \eqref{eq:val_nonbab} reduces to $\hat{V}(p_2 \mid p_1,b_1) \approx (L - c) \cdot \left( 1 + 1 - F_L (\bar{v}_2(p_2))\right)$. This expression is negative because $L < c$, whereas sending $m_0$ yields $V(m_0 \mid p_1,b_1) = 0$, and is therefore preferred. 
	
	Consider now ``the weakest recommendation to buy'' -- the smallest posterior belief among those available in equilibrium that lead $C_2$ to purchase the product, $m_1 = \min \{ m | m \in \mathcal{E}_1(p_1)\}$. After which experiences $v_1$ or, equivalently, for which private posteriors $b_1$ will $C_1$ send this message? Let $\varDelta_1$ denote the private posterior such that $C_1$ is indifferent between sending reviews $m_0$ and $m_1$: 
	\begin{align*}
		V_1(m_0 \mid p_1, \varDelta_1) &= V_1(m_1 \mid p_1, \varDelta_1) &&\iff&& 0=\hat{V}(m_1 \mid p_1, \varDelta_1).
	\end{align*}
	Since $\hat{V}(p_2 \mid p_1, b_1)$ is (linearly) increasing in $b_1$ for a given $p_2$, it follows that $C_1$ will prefer to send $m_0$ when $b_1 < \varDelta_1$ and to send $m_1$ when $b_1 > \varDelta_1$. Further, recall from Section \ref{sub:EX1p} that $\hat{V}(p_2 \mid p_1, \bar{p}) > 0$ for all $p_2$, so $\varDelta_1 < \bar{p}$. Importantly, $C_1$ also prefers $m_1$ to all other reviews $m_j > m_1$ when $b_1 \in (\varDelta_1, \bar{p}]$, because $\hat{V}(p_2 \mid p_1,b_1)$ is decreasing in $p_2$ when $p_2 > b_1$ (due to single-peakedness shown in \ref{sub:EX1p}).
	In the end, there exists a non-trivial range $(\varDelta_1, \bar{p}]$ of private posteriors, for which $C_1$ believes the product yields negative expected consumption utility, but she finds it strictly optimal to send $m_1$ and recommend a purchase. The intuition is the same as in Section \ref{sec:EX}: $C_1$ wants $C_2$ to purchase the product in order to generate information for $C_3$, even if it is not myopically optimal for $C_2$.
	
	In order to understand what the whole equilibrium communication strategy looks like, consider the (potentially infinite) set $\mathcal{M}_1(p_1)$ of messages $m_0 < m_1 < m_2 < ...$ that are available to $C_1$ in equilibrium. Then single-peakedness of $\hat{V}(p_2 \mid p_1,b_1)$ implies that $C_1$'s optimal (up to indifference) communication strategy is pure and monotone: higher $b_1$ leads to a higher message $m$. Conversely, the range of private beliefs for which $m_1$ is sent has to be convex, call it $(\varDelta_1, \varDelta_2]$. 
	Since $C_2$ is rational and Bayesian, her inference from $m_1$ must be consistent with $C_1$'s equilibrium strategy: $m_1 = \mathbb{E} \left[b_1 \mid b_1 \in (\varDelta_1, \varDelta_2] \right]$. Since $m_1 \geqslant \bar{p}$ and $\varDelta_1 < \bar{p}$, this implies that $\varDelta_2 > m_1$.
	In turn, $C_1$ with posterior $b_1 = \Delta_2$ must (by continuity of $\hat{V}$) be indifferent between leaving review $m_1$ and a higher review $m_2 > m_1$. However, we know that $\hat{V}(p_2 \mid p_1,b_1)$ is single peaked in $p_2$ with a peak at $p_2 = b_1$, hence the indifference condition $\hat{V}(m_1 \mid p_1,\Delta_2) = \hat{V}(m_2 \mid p_1,\Delta_2)$ implies that $m_2 > \varDelta_2$. 
	By iterating the argument, we conclude that 
	\begin{equation*}
		... < \varDelta_j < m_j < \varDelta_{j+1} < m_{j+1} < ...
	\end{equation*}
	with strict inequalities at every step. This sequence may be either finite with $\varDelta_j = 1$ at some step, or infinite, depending on distributions $F_\theta$.
	Therefore, equilibrium communication at $t = 1$ necessarily has an interval structure: instead of communicating her private belief $b_1$ truthfully (or, equivalently, consumption utility $v_1$ she received), $C_1$ only indicates which interval $(\varDelta_j, \varDelta_{j+1}]$ her posterior $b_1$ belongs to.
	Intuitively, the fact that the aforementioned ``most cautious recommendation to buy'', $m_1$, is noisy and not perfectly revealing of $b_1$ implies that all other messages must be noisy as well. Notably, perfect communication is thus impossible even for high posteriors $b_1$ when there is no conflict between the sender and the receiver.
	The latter suggests that the classic result of \citet{CS}, who show that cheap talk communication has an interval structure, does not rely on the conflict being present throughout the whole state space, as their model assumes.
	
	Figure \ref{fig:optimal_d} illustrates the payoffs in a potential cheap talk equilibrium. It plots the continuation payoff of $C_1$ in an interval equilibrium with three messages, $m_0 < m_1 = \bar{p} < m_2$.\footnote{A cheap talk equilibrium like the one in Figure \ref{fig:optimal_d} may or may not exist, depending on the value distributions $F_\theta$. Specifically, equilibria with more or fewer messages on the equilibrium path may exist, and posterior $m_1 \geqslant \bar{p}$ may not necessarily be exactly equal to $\bar{p}$.} Panel (a) plots this payoff against the first-best value $\hat{V}(b_1 \mid p_1, b_1)$ (obtained if $C_1$ could relay the correct belief $b_1$ \emph{and} force $C_2$ to purchase the item). The two coincide whenever $b_1 \in \{ m_1, m_2 \}$, but the cheap talk value is strictly lower for all other posteriors. This is due to the fact that while $C_1$ manages to convince $C_2$ to experiment with the product, the noise in communication makes the purchasing decision of the \emph{third} consumer less efficient.
	
	Panel (b) of Figure \ref{fig:optimal_d} compares a cheap talk equilibrium to the optimal social norm. While it is obvious that cheap talk can be no better than if $C_1$ had the ability to commit, the figure shows where these losses come from. Specifically, it is the noise for high $b_1$ that harms welfare -- in those cases, there is no conflict between $C_1$ and $C_2$, so truthful communication would be optimal, but the noise propagates from the ``conflict region'' $b_1 \in (\varDelta_1, \bar{p}]$, as argued above. At the same time, we can see that there can be more experimentation under cheap talk: weak recommendation $m_1$ is sent after a possibly wider range of posteriors $b_1$, including some in the conflict region, for which the optimal norm prescribes sending a negative review, $m_0$. Under the optimal social norm, $C_1$ commits to not recommend particularly bad products, for which the social value of experimentation is marginally positive, in order to be able to communicate truthfully more often (i.e., for a wider range of $b_1$) when $b_1 > \bar{p}$. Under cheap talk, however, she cannot avoid the temptation to recommend such products, which leads to more noise and lower welfare overall.
	Finally, while the equilibrium in Figure \ref{fig:optimal_d} has the feature that $m_1 = \bar{p}$, it may well be the case that $m_1 > \bar{p}$ in equilibrium. If so, this adds to inefficiency by reducing experimentation for $b_1 < \bar{p}$ while not necessarily reducing the average noise for $b_1 > \bar{p}$.
	
	\begin{figure}
		\centering
		\subfloat[][Equilibrium communication under cheap talk.]{
			\begin{tikzpicture}[xscale=6.5,yscale=6, 
				vhat/.style={dashed},
				truth/.style={color=violet!70!white, line width=1},
				commit/.style={color=red!70!yellow, line width=1},
				cheap/.style={color=cyan, line width=1},
				]
				\plotvhattrue
				\plottruthfalse
				\plotcommitfalse
				\plotcheaptrue
				
				\def\offsettruth{0.005};
				\def\offsetcommit{0.01};
				\def\offsetcheap{0.004};
				
				\valuefig
				
				\filldraw[cyan] (\pbar,0.125-\offsetcheap) ellipse(0.01) node[above]{$m_1$};
				\def\m2{0.83};
				\filldraw[cyan] (\m2,\m2*\m2*\m2-\offsetcheap) ellipse(0.01) node[above left]{$m_2$};
				\draw (\cheapL,0) node[below]{\footnotesize $\Delta_1$};
				\draw (\cheapR,0) node[below]{\footnotesize $\Delta_2$};
				\draw[dotted] (\cheapR,0) -- (\cheapR,25/12*\cheapR-125/108-\offsetcheap);
				
				\def\stroke{++(-0.2,0) -- ++(0.4,0)};
				\matrix [draw, fill=white, below right] at (0.05,1.0) {
					\draw [vhat] \stroke node[black,right] {\footnotesize $\hat{V}(b_1 \mid p_1,b_1)$}; \\ 
					\draw [cheap] \stroke node[black,right] {\footnotesize $V_1^{CT}(m^* \mid p_1,b_1)$}; \\
				};
			\end{tikzpicture}
		}
		\subfloat[][Cheap talk vs commitment.]{
			\begin{tikzpicture}[xscale=6.5,yscale=6, 
				vhat/.style={dashed},
				truth/.style={color=violet!70!white, line width=1},
				commit/.style={color=red!70!yellow, line width=1, dashed},
				cheap/.style={color=cyan, line width=1},
				]
				\plotvhatfalse
				\plottruthfalse
				\plotcommittrue
				\plotcheaptrue
				
				\def\offsettruth{0.003};
				\def\offsetcommit{0.0};
				\def\offsetcheap{0.005};
				
				\valuefig
				
				\def\stroke{++(-0.2,0) -- ++(0.4,0)};
				\matrix [draw, fill=white, below right] at (0.05,1.0) {
					\draw [commit] \stroke node[black,right] {\footnotesize $V_1^{Comm}(m^* \mid p_1,b_1)$}; \\
					\draw [cheap] \stroke node[black,right] {\footnotesize $V_1^{CT}(m^* \mid p_1,b_1)$}; \\
				};
			\end{tikzpicture}
		}
		\caption{$C_1$'s continuation value with cheap talk. \label{fig:optimal_d}}
		
		\footnotesize\emph{Note: Solid blue line in both panels plots $C_1$'s value in a cheap talk equilibrium; dashed red line in panel (b) plots $C_1$'s value from optimal communication with commitment; $m^*$ denotes the messages sent by the optimal communication strategies in the respective scenarios.}
	\end{figure}

	\subsection{Infinite-Horizon Model} \label{sub:infct}
	
	We now generalize the intuition of the three-period model with cheap talk to an infinite horizon. We again focus on stationary MPE with Cheap Talk, in which the communication strategies $\mu (m | p_t, b_t)$ and belief updating rules $q(p_t,m)$ are time-invariant. The equilibrium multiplicity problem associated with cheap talk models is greatly exacerbated with an infinite horizon, even with a restriction to stationary equilibria. This prevents us from providing a tight equilibrium characterization. However, we can still show that there must be a pooling region around $\bar{p}$, unless communication in all subsequent histories is sufficiently uninformative. In the latter case, we argue that truthful communication can arise, which is not trivial with an infinite horizon, since a ``terminal consumer'' no longer exists. 
	
	To formulate the results, we first introduce the notion of a cascade from the observational learning literature (\citet*{BHTW}), where it describes situations when social learning stops and the society locks in on one alternative (possibly the wrong one).
	
	\begin{definition*}
		Message $m \in \mathcal{E}(p_t)$ at public state $p_t$ starts a \emph{cascade} if after such a message, $p_s \geqslant \bar{p}$ for all $s > t$.
	\end{definition*}

	In other words, we say that some recommendation to purchase issued at $p_t$ leads to all future consumers buying the product, regardless of any of the interim consumers' experiences and reviews. Once a cascade starts, no new reviews can change future consumers' behavior. There are two things to note in relation to cascades. First, any message $m \in \mathcal{S}(p_t)$ at any $p_t$ necessarily starts a cascade as well, in the sense that no future consumers buy the product again, as discussed in Section \ref{sub:hist}.
	Second, with cheap talk, a continuation equilibrium always exists in which any given $m \in \mathcal{E}(p_t)$ starts a cascade. One example is the babbling equilibrium, one in which all future reviews are uninformative and are perceived as such, and thus the public belief remains frozen at $q(m | p_t)$.\footnote{
		Babbling equilibria are ubiquitous in cheap talk models. To see that babbling is an equilibrium, note that neither player has a profitable deviation. The sender cannot benefit by sending informative messages because they are ignored by the receivers regardless, and the receivers cannot benefit by following the sender's recommendation since it is uninformative. 
	}
	However, in general, a cascade need not shut down information transmission completely: reviews may be informative and affect the public belief $p_t$ as long as they do not affect future consumers' actual purchasing decisions.
	
	\begin{proposition} \label{prop:NPR_CT}
		In any stationary MPE with cheap talk, in any public state $p_t$: $[\bar{p},1] \subset \mathcal{P}(p_t)$ only if any message $m \in \mathcal{M}(p_t)$ starts a cascade.
	\end{proposition}
	
	Proposition \ref{prop:NPR_CT} demonstrates that the conflict between the sender and the receiver of a review precludes perfect communication. It claims that unless \emph{any} message $m$ available in period $t$ starts a cascade, $C_t$ cannot have access to all possible public posteriors $[\bar{p}, 1]$. This implies that some information is inevitably lost in period $t$, unless all information after period $t$ is ignored. The idea is that if $C_{t+1}$ can provide information for future consumers with her purchase, then $C_t$ wants her to produce this externality, which leads $C_t$ to misreport her posterior in some cases. Conversely, if no informative communication is possible at $t+1$ or afterwards, then $C_t$ has no reason to induce experimentation that $C_{t+1}$ is trying to avoid, because the information generated by $C_{t+1}$ would not be meaningful to subsequent consumers either way. In the latter case, $C_t$ can be completely truthful with $C_{t+1}$, just as $C_2$ could be truthful with $C_3$ in the three-period example.

	Proposition \ref{prop:NPR_CT} is a negative statement, claiming that perfectly informative equilibria are unattainable in the cheap talk game beyond a single period. Theorem \ref{thm:dec} below is, conversely, a positive statement, providing a partial (if weak) characterization of what informative equilibria \emph{must} look like.
	
	\begin{theorem} \label{thm:dec}
		In any stationary MPE with cheap talk, for any $p_t$ for which there exists a message that does not start a cascade, there exist $l(p_t)$ and $r(p_t)$ such that $l(p_t) < \bar{p} < r(p_t)$, and for all $b_t \in [l(p_t), r(p_t)]$ we have $\mu (m \mid p_t,b_t) = 1$ for one such $m$.
	\end{theorem}
	
	Theorem \ref{thm:dec} claims that except in the cascade scenario discussed above, experiences $b_t$ in some neighborhood of the myopic cutoff $\bar{p}$ are always pooled together into a single review. 
	The intuition mirrors that from Sections \ref{sec:EQ} and \ref{sub:int3}: if one of the consumers arriving after $t$ is able to leave informative reviews, the option value of this information makes it optimal for $C_t$ to make $C_{t+1}$ buy the product when it is myopically suboptimal for the latter. To make this recommendation credible, it must also sometimes be issued when $b_t > \bar{p}$.
	The part that is worth pointing out here is the qualifier on $p_t$: communication at $p_t$ must be noisy only if at least some message is available in $\mathcal{M}(p_t)$ that does not start a cascade -- i.e., if at least some future consumer can issue a pivotal review. The complementary case was discussed in Proposition \ref{prop:NPR_CT}: if all messages in $\mathcal{M}(p_t)$ start a cascade, then perfect communication in state $p_t$ is possible.

	\section{Conclusion} \label{sec:CON}
	
	This paper presents a theoretical model of social learning, focusing on the issue of information provision in product reviews. We look closely at the fundamental tension underlying product reviews -- the conflict between consumers' self-interest in purchasing behavior and their prosocial motives when writing reviews -- and investigate how this tension affects the informational content of the reviews. We show that even in the absence of external interference (from, e.g., a platform or the sellers), this tension inevitably leads to a breakdown of truthful communication, as reviewers desire to deceive future consumers into buying a potentially subpar product for the sake of generating information. Moreover, despite the conflict only arising under specific circumstances, the noise created by it can propagate, making \emph{all} communication noisy in equilibrium. 
	
	Furthermore, our model connects two well-documented empirical patterns: the prevalence of prosocial motives in review-writing, and the tendency for consumers to leave reviews primarily after extremely positive or negative experiences. We show that leaving noisy reviews (or abstention from reviewing altogether) after average experiences improves welfare relative to truthful communication, because it encourages social experimentation. Therefore, this kind of behavior can emerge as an optimal strategy in the presence of prosocial motives. Alternatively, non-reporting of average experiences can arise as the optimal social norm in the society, which is not necessarily a result of individual consumers' rational intent.
	
	Beyond the specific context of product reviews, our paper also contributes to the broader literature on social learning and information aggregation. For instance, in the realm of scientific research, our results suggest that the non-publication of ``weak'' or null results may be optimal from a social welfare perspective. While it may lead to the duplication of research in the short term, it may also stimulate the production of stronger, more convincing findings in the long run, improving social welfare overall.

	\bibliographystyle{abbrvnat}
	\bibliography{literature}

\begin{thebibliography}{47}
\providecommand{\natexlab}[1]{#1}
\providecommand{\url}[1]{\texttt{#1}}
\expandafter\ifx\csname urlstyle\endcsname\relax
  \providecommand{\doi}[1]{doi: #1}\else
  \providecommand{\doi}{doi: \begingroup \urlstyle{rm}\Url}\fi

\bibitem[Ali and Kartik(2012)]{AK}
S.~N. Ali and N.~Kartik.
\newblock Herding with collective preferences.
\newblock \emph{Economic Theory}, 51\penalty0 (3):\penalty0 601--626, November
  2012.

\bibitem[Alonso and Camara(2016)]{AC}
R.~Alonso and O.~Camara.
\newblock Bayesian persuasion with heterogeneous priors.
\newblock \emph{Journal of Economic Theory}, 165:\penalty0 672--706, 2016.

\bibitem[Ambrus et~al.(2013)Ambrus, Azevedo, and Kamada]{AAK}
A.~Ambrus, E.~Azevedo, and Y.~Kamada.
\newblock Hierarchical cheap talk.
\newblock \emph{Theoretical Economics}, 8\penalty0 (1):\penalty0 233--261,
  January 2013.

\bibitem[Avery et~al.(1999)Avery, Resnick, and Zeckhauser]{ARZ}
C.~Avery, P.~Resnick, and R.~Zeckhauser.
\newblock The market for evaluations.
\newblock \emph{American Economic Review}, 89\penalty0 (3):\penalty0 564--584,
  June 1999.

\bibitem[Barto\v{s} et~al.(2024)Barto\v{s}, Maier, Wagenmakers, Nippold,
  Doucouliagos, Ioannidis, Otte, Sladekova, Deresssa, Bruns, Fanelli, and
  Stanley]{BMWetal}
F.~Barto\v{s}, M.~Maier, E.~Wagenmakers, F.~Nippold, H.~Doucouliagos, J.~P.~A.
  Ioannidis, W.~M. Otte, M.~Sladekova, T.~K. Deresssa, S.~B. Bruns, D.~Fanelli,
  and T.~D. Stanley.
\newblock Footprint of publication selection bias on meta-analyses in medicine,
  environmental sciences, psychology, and economics.
\newblock \emph{Research Synthesis Methods}, 15\penalty0 (3):\penalty0
  500--511, February 2024.

\bibitem[B{\'e}nabou and Vellodi(2024)]{BV24}
R.~B{\'e}nabou and N.~Vellodi.
\newblock ({P}ro-)social learning and strategic disclosure.
\newblock Working paper, 2024.

\bibitem[Bik(2024)]{Bik}
E.~M. Bik.
\newblock Publishing negative results is good for science.
\newblock \emph{Access Microbiology}, 6\penalty0 (4):\penalty0 000792, 2024.

\bibitem[Bikhchandani et~al.(2024)Bikhchandani, Hirshleifer, Tamuz, and
  Welch]{BHTW}
S.~Bikhchandani, D.~Hirshleifer, O.~Tamuz, and I.~Welch.
\newblock Information cascades and social learning.
\newblock \emph{Journal of Economic Literature}, 62\penalty0 (3):\penalty0
  1040--1093, September 2024.

\bibitem[Blanco-Perez and Brodeur(2020)]{BPB20}
C.~Blanco-Perez and A.~Brodeur.
\newblock Publication bias and editorial statement on negative findings.
\newblock \emph{The Economic Journal}, 130\penalty0 (629):\penalty0 1226--1247,
  2020.

\bibitem[Bolton and Harris(1999)]{BH}
P.~Bolton and C.~Harris.
\newblock Strategic experimentation.
\newblock \emph{Econometrica}, 67\penalty0 (2):\penalty0 349--374, March 1999.

\bibitem[Cao et~al.(2011)Cao, Han, and Hirshleifer]{CBH}
H.~H. Cao, B.~Han, and D.~Hirshleifer.
\newblock Taking the road less traveled by: Does conversation eradicate
  pernicious cascades?
\newblock \emph{Journal of Economic Theory}, 146\penalty0 (4):\penalty0
  1418--1436, July 2011.

\bibitem[Carnehl and Schneider(2024)]{CS24}
C.~Carnehl and J.~Schneider.
\newblock A quest for knowledge.
\newblock Working Paper, 2024.

\bibitem[Che and H{\"o}rner(2018)]{CH}
Y.-K. Che and J.~H{\"o}rner.
\newblock Recommender systems as mechanisms for social learning.
\newblock \emph{Quarterly Journal of Economics}, 133\penalty0 (2):\penalty0
  871--925, May 2018.

\bibitem[Chiba(2018)]{Chi}
S.~Chiba.
\newblock Hidden profiles and persuasion cascades in group decision-making.
\newblock Working Paper, 2018.

\bibitem[Cohen and Mansour(2019)]{CM}
L.~Cohen and Y.~Mansour.
\newblock Optimal algorithm for bayesian incentive-compatible exploration.
\newblock In \emph{Proceedings of the 2019 ACM Conference on Economics and
  Computation}, pages 135--151, 2019.

\bibitem[{Competition \& Markets Authority}(2015)]{CMA}
{Competition \& Markets Authority}.
\newblock Online reviews and endorsements.
\newblock
  \url{https://www.gov.uk/government/uploads/system/uploads/attachment_data/file/436238/Online_reviews_and_endorsements.pdf/},
  2015.
\newblock Retrieved January, 2020.

\bibitem[Crawford and Sobel(1982)]{CS}
V.~Crawford and J.~Sobel.
\newblock Strategic information transmission.
\newblock \emph{Econometrica}, 50\penalty0 (6):\penalty0 1431--1451, November
  1982.

\bibitem[Dye(1985)]{DYE}
R.~A. Dye.
\newblock Disclosure of nonproprietary information.
\newblock \emph{Journal of Accounting Research}, 23\penalty0 (1):\penalty0
  123--145, Spring 1985.

\bibitem[eMarketer(2018)]{eMarketer}
eMarketer.
\newblock Surprise! most consumers look at reviews before a purchase, 2018.
\newblock URL
  \url{https://www.emarketer.com/content/surprise-most-consumers-look-at-reviews-before-a-purchase}.
\newblock Retrieved January, 2020.

\bibitem[Evans and Gariepy(2015)]{EG}
L.~C. Evans and R.~F. Gariepy.
\newblock \emph{Measure Theory and Fine Properties of Functions, Revised
  Edition}.
\newblock CRC Press, 2015.

\bibitem[Falk and Szech(2013)]{FS13}
A.~Falk and N.~Szech.
\newblock Morals and markets.
\newblock \emph{Science}, 340\penalty0 (6133):\penalty0 707--711, 2013.

\bibitem[Fehr and Schmidt(2003)]{FS2}
E.~Fehr and K.~M. Schmidt.
\newblock Theories of fairness and reciprocity: Evidence and economic
  application.
\newblock \emph{M. Dewatripont, L.P. Hansen, and S.J. Turnovsky, Advances in
  Economics and Econometrics -- 8th World Congress, Econometric Society
  Monographs}, pages 208--257, 2003.

\bibitem[Guarino et~al.(2011)Guarino, Harmgart, and Huck]{GHH}
A.~Guarino, H.~Harmgart, and S.~Huck.
\newblock Aggregate information cascades.
\newblock \emph{Games and Economic Behavior}, 73\penalty0 (1):\penalty0
  167--185, 2011.

\bibitem[Heidhues et~al.(2015)Heidhues, Rady, and Strack]{HRS}
P.~Heidhues, S.~Rady, and P.~Strack.
\newblock Strategic experimentation with private payoffs.
\newblock \emph{Journal of Economic Theory}, 159:\penalty0 531--551, 2015.

\bibitem[Inostroza and Pavan(2023)]{IP}
N.~Inostroza and A.~Pavan.
\newblock Adversarial coordination and public information design.
\newblock Working paper, 2023.

\bibitem[Keller et~al.(2005)Keller, Rady, and Cripps]{KRC}
G.~Keller, S.~Rady, and M.~Cripps.
\newblock Strategic experimentation with exponential bandits.
\newblock \emph{Econometrica}, 73\penalty0 (1):\penalty0 39--68, January 2005.

\bibitem[Konow(2003)]{Kon}
J.~Konow.
\newblock Which is the fairest one of all? a positive analysis of justice
  theories.
\newblock \emph{Journal of Economic Literature}, 41\penalty0 (4):\penalty0
  1188--1239, December 2003.

\bibitem[Kremer et~al.(2014)Kremer, Mansour, and Perry]{KMP}
I.~Kremer, Y.~Mansour, and M.~Perry.
\newblock Implementing the ``wisdom of the crowd''.
\newblock \emph{Journal of Political Economy}, 122\penalty0 (5):\penalty0
  988--1012, October 2014.

\bibitem[Le et~al.(2016)Le, Subramanian, and Berry]{LSB}
T.~N. Le, V.~G. Subramanian, and R.~A. Berry.
\newblock Are imperfect reviews helpful in social learning?
\newblock In \emph{Proceedings of the 2016 IEEE International Symposium on
  Information Theory}, pages 2089--2093, July 2016.

\bibitem[Liang and Mu(2020)]{LM}
A.~Liang and X.~Mu.
\newblock Complementary information and learning traps.
\newblock \emph{The Quarterly Journal of Economics}, 135\penalty0 (1):\penalty0
  389--448, 2020.

\bibitem[Luca and Zervas(2016)]{LZ}
M.~Luca and G.~Zervas.
\newblock Fake it till you make it: Reputation, competition, and yelp review
  fraud.
\newblock \emph{Management Science}, 62\penalty0 (12):\penalty0 3412--3427,
  December 2016.

\bibitem[Mansour et~al.(2015)Mansour, Slivkins, and Syrgkanis]{MSS}
Y.~Mansour, A.~Slivkins, and V.~Syrgkanis.
\newblock Bayesian incentive-compatible bandit exploration.
\newblock In \emph{Proceedings of the Sixteenth ACM Conference on Economics and
  Computation}, pages 565--582, 2015.

\bibitem[March and Ziegelmeyer(2020)]{MZ}
C.~March and A.~Ziegelmeyer.
\newblock Altruistic observational learning.
\newblock \emph{Journal of Economic Theory}, 190:\penalty0 105--123, 2020.

\bibitem[Mariotti et~al.(2023)Mariotti, Schweizer, Szech, and von
  Wangenheim]{MSSW}
T.~Mariotti, N.~Schweizer, N.~Szech, and J.~von Wangenheim.
\newblock Information nudges and self-control.
\newblock \emph{Management Science}, 69\penalty0 (4):\penalty0 2182--2197,
  2023.

\bibitem[Meier(2006)]{Mei}
S.~Meier.
\newblock A survey of economic theories and field evidence on pro-social
  behavior.
\newblock Working Paper, 2006.

\bibitem[Mintel(2015)]{Mintel}
Mintel.
\newblock Seven in 10 americans seek out opinions before making purchases,
  2015.
\newblock URL
  \url{https://www.mintel.com/press-centre/social-and-lifestyle/seven-in-10-americans-seek-out-opinions-before-making-purchases}.
\newblock Retrieved January, 2020.

\bibitem[Nimpf and Keays(2020)]{NK20}
S.~Nimpf and D.~A. Keays.
\newblock Why (and how) we should publish negative data.
\newblock \emph{EMBO reports}, 21\penalty0 (1):\penalty0 e49775, 2020.

\bibitem[Peng et~al.(2017)Peng, Rao, Sun, and Xiao]{PRSX}
D.~Peng, Y.~Rao, X.~Sun, and E.~Xiao.
\newblock Optional disclosure and observational learning.
\newblock Working paper, 2017.

\bibitem[Renault et~al.(2013)Renault, Solan, and Vieille]{RSV}
J.~Renault, E.~Solan, and N.~Vieille.
\newblock Dynamic sender--receiver games.
\newblock \emph{Journal of Economic Theory}, 148\penalty0 (2):\penalty0
  502--534, 2013.

\bibitem[Smirnov(2020)]{Smir20}
A.~Smirnov.
\newblock \emph{Essays on social learning and reputation building}.
\newblock PhD thesis, University of Zurich, 2020.

\bibitem[Smirnov and Starkov(2022)]{SScens}
A.~Smirnov and E.~Starkov.
\newblock Bad news turned good: Reversal under censorship.
\newblock \emph{American Economic Journal: Microeconomics}, 14\penalty0
  (2):\penalty0 506--560, 2022.

\bibitem[Smith and S{\o}rensen(2000)]{SS_2000}
L.~Smith and P.~S{\o}rensen.
\newblock Pathological outcomes of observational learning.
\newblock \emph{Econometrica}, 68\penalty0 (2):\penalty0 371--398, March 2000.

\bibitem[Smith et~al.(2021)Smith, S{\o}rensen, and Tian]{SST}
L.~Smith, P.~S{\o}rensen, and J.~Tian.
\newblock Informational herding, optimal experimentation, and contrarianism.
\newblock \emph{Review of Economic Studies}, 88\penalty0 (5):\penalty0
  2527--2554, 2021.

\bibitem[Stokey et~al.(1989)Stokey, Lucas, and Prescott]{SLP}
N.~L. Stokey, R.~E. Lucas, and E.~C. Prescott.
\newblock \emph{Recursive Methods in Economic Dynamics}.
\newblock Harvard University Press, 1989.

\bibitem[Swank and Visser(2015)]{SV15}
O.~Swank and B.~Visser.
\newblock Learning from others? decision rights, strategic communication, and
  reputational concerns.
\newblock \emph{American Economic Journal: Microeconomics}, 7\penalty0
  (4):\penalty0 109--149, 2015.

\bibitem[Trustpilot(2020)]{Trustpilot}
Trustpilot.
\newblock Why do people write reviews? what our research revealed, 2020.
\newblock URL
  \url{https://business.trustpilot.com/reviews/learn-from-customers/why-do-people-write-reviews-what-our-research-revealed}.
\newblock Retrieved April, 2022.

\bibitem[Wolitzky(2018)]{Wol}
A.~Wolitzky.
\newblock Learning from others' outcomes.
\newblock \emph{American Economic Review}, 108\penalty0 (10):\penalty0
  2763--2801, October 2018.

\end{thebibliography}

	\newpage

	\section*{Appendix}
	\renewcommand{\baselinestretch}{1.33}\small
	\renewcommand\thesubsection{\Alph{section}.\arabic{subsection}}
	\setcounter{section}{1}

	\subsection{Supplementary Lemmas}
	
	In this section, we introduce a number of supplementary results that will aid us in proving Theorem \ref{thm:commitment}. Therefore, ``equilibrium'' should be understood as ``stationary MPE with Commitment''.
	
	First, fix an arbitrary state $p_t \in [\bar{p}, 1)$.\footnote{Remember that if $p_t < \bar{p}$ the experimentation stops, while if $p_t = 1$ the distribution becomes degenerate and reduces to a point mass at $b=1$.} 
	Let $\Phi_\theta(b | p_t)$ denote the c.d.f. of the private posterior $b_t$ from $C_t$'s point of view given prior $p_t$ and true state $\theta$. We next derive the exact expression for $\Phi_\theta(b | p_t)$.
	\begin{lemma} \label{lem:F_distribution}
	The distribution of $b_t$ conditional on $\theta$ and public belief $p_t$ is
	\begin{equation*} 
 		\Phi_\theta(b | p_t) = F_\theta \left(L^{-1} \left(\ln \frac{b}{1-b} - \ln \frac{p_t}{1 - p_t}\right)\right).
	\end{equation*}
	\end{lemma}
	\begin{proof}
		  By definition
			\begin{align*} 
			& \Phi_\theta(b|p_t) = \mathbb{P} \left\{ b(p_t,v_t) \leqslant b \mid p_t, \theta\right\} = \mathbb{P} \left\{ \ln \frac{b(p_t,v_t)}{1 - b(p_t,v_t)} \leqslant \ln \frac{b}{1-b} \mid p_t, \theta \right\} = \\
			& = \mathbb{P} \left\{ \ln \frac{p_t}{1 - p_t} + L(v_t) \leqslant \ln \frac{b}{1-b} \mid p_t, \theta \right\} = \mathbb{P} \left\{ v_t \leqslant L^{-1} \left(\ln \frac{b}{1-b} - \ln \frac{p_t}{1 - p_t}\right) \mid p_t, \theta \right\} = \\
			& = F_\theta \left(L^{-1} \left(\ln \frac{b}{1-b} - \ln \frac{p_t}{1 - p_t}\right)\right). \qedhere
		  \end{align*}
	\end{proof}
	
	Given $\Phi_\theta(b | p_t)$, we can introduce $\Phi(b | p_t)$ to denote the c.d.f. of $C_t$'s private posterior $b_t$ given prior $p_t$.
	\begin{equation} \label{eq:phi}
		\Phi(b|p_t) = p_t \cdot \Phi_H(b | p_t) + (1-p_t) \cdot \Phi_L(b | p_t).
	\end{equation}
	Note that the MLRP assumption implies that $L'(v) > 0$ for any $v$, and therefore the respective p.d.f. $\phi(b|p_t) := \frac{d\Phi(b|p_t)}{db}$ is well defined.
	Then $q(m | p_t, \mu)$ can in equilibrium be written as an expectation with respect to $\Phi(b|p_t)$ of the private posterior $b_t$ of the agent who sent message $m$:
		\begin{equation} \label{eq:q_b}
		q(m | p_t) = \mathbb{E} [b | m,p_t] = \frac{\int\limits_0^1 b \cdot \mu(m|p_t,b) \cdot d \Phi(b|p_t)}{\int\limits_0^1 \mu(m|p_t,b) \cdot d \Phi(b|p_t)}.
	\end{equation}
		
	Without loss, hereafter, we assume direct communication strategies, where $m = q(m | p_t)$. The next result provides a more convenient representation for value function \eqref{eq:value}. It shows that the continuation value $V(m \mid p_t,b_t)$ after any message $m \in \mathcal{E}(p_t)$ can be characterized by two value functions $V^H(m)$ and $V^L(m)$ that do not depend on $b_t$ or $p_t$.
	\begin{lemma} \label{lem:representation}
		Fix any equilibrium.
		For any public belief $p_t$, if $m \in \mathcal{E}(p_t)$ then 
		\begin{equation} \label{eq:representation}
		  V(p_{t+1} \mid p_t,b_t) = \theta(b_t) - c + \beta \cdot \Big[b_t \cdot V^H(p_{t+1}) + (1-b_t) \cdot V^L(p_{t+1})\Big],
		\end{equation}
		where $p_{t+1} = m = q(m | p_t)$ and
		\begin{equation} \label{eq:vtheta}
			V_\theta(p_t) := \mathbb{E} \left[ \sum\limits_{s=t+1}^{+\infty} \beta^{s-t-1} \cdot \mathbb{I}\left(p_s \geqslant \bar{p}\right) \cdot (v_s - c) \: \bigg| \: p_t, \theta \right].
		\end{equation}
		Moreover, $V(m \mid p_t,b_t)$ is linear and strictly increasing in $b_t$ for any given $m \in \mathcal{E}(p_t)$.
	\end{lemma}
	\begin{proof}
		Assumption $m \in \mathcal{E}(p_t)$ means $C_{t+1}$ will buy the product, receiving a payoff that $C_t$ estimates at $\mathbb{E}[v_{t+1} | b_t] = \theta(b_t)$. For all $s \geqslant t+2$, $v_{s}$ are independent from $\{v_1,\dots,v_t\}$, so $p_s$ are independent from $b_t$ given $m$. Hence \eqref{eq:value} reduces to exactly \eqref{eq:representation}.
		Linearity follows immediately, since both $\theta(b_t)$ and $b_t \cdot V^H(p_t) + (1-b_t) \cdot V^L(p_t)$ are linear in $b_t$. 
		To show monotonicity, observe first that $\theta(b_t)$ is strictly increasing in $b_t$. Second, for any $s \geqslant t+1$: $\mathbb{E}\left[ v_s - c \mid p_t, \theta=H \right] > 0 > \mathbb{E}\left[ v_s - c \mid p_t, \theta=L \right]$, and $v_s$ is independent of $p_s$ (because $v_s$ is independent of $\{v_1,\dots,v_{s-1}\}$), meaning that $V^H(p_{t+1}) \geqslant 0 \geqslant V^L(p_{t+1})$. The result then follows.
	\end{proof}

	\begin{lemma} \label{lem:eps}
		Fix an arbitrary $p_t \in [\bar{p}, 1)$ and the corresponding distribution $\Phi(b|p_t)$ given by \eqref{eq:phi}. Suppose $b_t \sim \Phi(b|p_t)$.
		For any $l \in [0,\bar{p}]$, let $r(l)$ be implicitly defined by the following condition:
		\begin{equation}
			\mathbb{E} \left[b \mid b \in [l,r]\right] = \bar{p}.
		\label{eq:rofl}
		\end{equation}
		Then
		\begin{enumerate}
			\item $r(l,p_t)$ is well defined for all $l \in [0, \bar{p}]$,
			\item for any $\varepsilon \in [0, \bar{p})$ there exist $\underline{\delta} (\varepsilon, p_t), \overline{\delta} (\varepsilon, p_t) > 0$ such that if $l \in [\bar{p} - \varepsilon, \bar{p}]$ then $r(l) - \bar{p} \in [\underline{\delta} \cdot (\bar{p} - l), \overline{\delta} \cdot (\bar{p} - l)]$,
			\item $r(l,p_t)$ is a continuous function on $[0, \bar{p}] \times [\bar{p}, 1)$.
		\end{enumerate}
	\end{lemma}
	\begin{proof}
		Rewriting definition \eqref{eq:rofl}, we get
		\begin{equation} \label{eq:rofl2}
			\int \limits_{\bar{p}}^{r} (b-\bar{p}) \cdot \phi (b|p_t) \; db = \int \limits_{l}^{\bar{p}} (\bar{p}-b) \cdot \phi (b|p_t) \; db.
		\end{equation}
		The left-hand side of \eqref{eq:rofl2} is continuous, strictly increasing in $r$, and equals zero for $r = \bar{p}$.
		Since $\mathbb{E}[b|p_t] = p_t \geqslant \bar{p}$, it follows that 
		\begin{align} \label{eq:rofl3}
			\int \limits_{\bar{p}}^{1} (b-\bar{p}) \cdot \phi (b|p_t) \; db &\geqslant \int \limits_{0}^{\bar{p}} (\bar{p}-b) \cdot \phi (b|p_t) \; db \geqslant \int \limits_{l}^{\bar{p}} (\bar{p}-b) \cdot \phi (b|p_t) \; db.
		\end{align}
		Combining the observations above, by the intermediate value theorem we conclude that $r(l,p_t)$ exists.
		
		Fix some $\varepsilon > 0$. Let $B_l := \min \{ \phi(b|p_t) \mid b \in [\bar{p} - \varepsilon, r(\bar{p}-\varepsilon)] \}$ and $B_h := \max\{ \phi(b|p_t) \mid b \in [\bar{p} -\varepsilon, r(\bar{p} - \varepsilon)] \}$. These bounds exist because $\phi (b|p_t)$ is continuous, strictly positive, and finite for any $p_t \in [\bar{p}, 1)$.\footnote{Continuity follows from Lemma \ref{lem:F_distribution} and the assumption that $f_L$ and $f_H$ are continuously differentiable. It implies that $\phi (b|p_t)$ is continuous because it can be expressed as a function of $f_L$, $f_H$ and $f'_L$, $f'_H$, which are all continuous.} Therefore we can obtain from \eqref{eq:rofl2} that
		\begin{align*}
			\frac{B_l}{2} \cdot \left( r(l) - \bar{p} \right)^2 &\leqslant \frac{B_h}{2} \cdot \left( \bar{p}-l \right)^2
			& and &&
			\frac{B_h}{2} \cdot \left( r(l)-\bar{p} \right)^2 &\geqslant \frac{B_l}{2} \cdot \left( \bar{p}-l \right)^2.			
		\end{align*}
		Recall that $B_l, B_h > 0$ (since $\phi(b|p_t) > 0$ for $b \in [\bar{p} - \varepsilon, r(\bar{p}-\varepsilon)] \subset (0,1)$), so by setting $\underline{\delta} := \sqrt{\frac{B_l}{B_h}}$ and $\overline{\delta} := \sqrt{\frac{B_h}{B_l}}$ we get the result.
		
		It remains to show the continuity of $r(l,p_t)$. Fix some $\varepsilon > 0$. Take some sequence $(l^n, p_t^n) \in [0, \bar{p}] \times [\bar{p}, 1-\varepsilon]$ that converges to a given $(l,p_t)$. Denote associated solutions to \eqref{eq:rofl2} as $r^n := r(l^n, p_t^n)$ and $r := r(l,p_t)$. Then
		\begin{equation} \label{eq:rofl22}
		  \int \limits_{\bar{p}}^{r^n} (b-\bar{p}) \cdot \phi (b|p^n_t) \; db = \int \limits_{l^n}^{\bar{p}} (\bar{p}-b) \cdot \phi (b|p^n_t) \; db.
		\end{equation}
		Let $n \rightarrow +\infty$. Because $p^n_t \in [\bar{p}, 1-\varepsilon]$ and $b \in [0,\bar{p}]$, it implies that $\phi (b|p^n_t)$ is uniformly bounded from above. Then by the Dominated Convergence Theorem (see \citet*{EG}, Theorem 1.19) the RHS in \eqref{eq:rofl22} converges to $\int \limits_{l}^{\bar{p}} (\bar{p}-b) \cdot \phi (b|p_t) \; db$ as a composition of continuous functions. Given \eqref{eq:rofl2} it implies
		\begin{equation*}
		  \lim\limits_{n \rightarrow +\infty} \int \limits_{r}^{r^n} (b-\bar{p}) \cdot \phi (b|p^n_t) \; db + \lim\limits_{n \rightarrow +\infty} \int \limits_{\bar{p}}^{r} (b-\bar{p}) \cdot \phi (b|p^n_t) \; db = \int \limits_{\bar{p}}^{r} (b-\bar{p}) \cdot \phi (b|p_t) \; db.
		\end{equation*}
		By the Dominated Convergence Theorem the second summand in the LHS converges to the RHS, and therefore
		\begin{equation*}
		  \lim\limits_{n \rightarrow +\infty} \int \limits_{r}^{r^n} (b-\bar{p}) \cdot \phi (b|p^n_t) \; db = 0.
		\end{equation*}
		Assume now the contrary, that is $r^n \nrightarrow r$. Then the measure of the integration area is separated from zero. Both functions within the integral are then also separated from zero on a non-zero measure set. Therefore, the integral above is separated from zero as well and can not converge to zero, which gives us a contradiction. The claim was established for any $\varepsilon > 0$ and therefore $r(l,p_t)$ is continuous on $[0,\bar{p}] \times [\bar{p}, 1)$. \qedhere
	\end{proof}

	\begin{lemma} \label{lem:mon_con}
		Suppose $f(x)$ is a [weakly] convex and differentiable function on $\left[\bar{p}, 1\right]$, $a, b > 0$ and $f(\bar{p}) = a \bar{p} + b$, $f'(\bar{p}) = a$. Then $\frac{f(x)}{ax+b}$ is a [weakly] increasing function on $\left[\bar{p}, 1\right]$.
	\end{lemma}
	\begin{proof}
		Consider $y > x \geqslant \bar{p}$. Then
		\begin{equation*}
			\frac{f(y)}{ay+b} - \frac{f(x)}{ax+b} = \frac{y-x}{(ay+b)(ax+b)} \cdot \left((ax+b) \cdot \frac{f(y) - f(x)}{y-x} - a \cdot f(x)\right).
		\end{equation*}
		Because $f(x)$ is convex we have that $\frac{f(y) - f(x)}{y-x} \geqslant \frac{f(x) - f(\bar{p})}{x - \bar{p}}$. Therefore,
		\begin{equation*}
			(ax+b) \cdot \frac{f(y) - f(x)}{y-x} - a \cdot f(x) \geqslant (ax+b) \cdot \frac{f(x) - f(\bar{p})}{x - \bar{p}} - a \cdot f(x) = (a\bar{p} + b) \cdot \frac{f(x) - f(\bar{p})}{x - \bar{p}} - a \cdot f(\bar{p}).
		\end{equation*}
		$f(\bar{p}) = a\bar{p} + b$ and using convexity of $f(x)$ once again we know that $\frac{f(x) - f(\bar{p})}{x - \bar{p}} \geqslant f'(\bar{p}) = a$. Therefore, the expression above is non-negative.
	\end{proof}

	\begin{lemma} \label{lem:neg}
		In any equilibrium, there exists $\Delta > 0$ such that for all $p_t \in [\bar{p},1)$, all $m \in \mathcal{E}(p_t)$, and all $b_t < \Delta$: $V(m \mid p_t, b_t) < 0$.
	\end{lemma}
	\begin{proof}
		Since $m \in \mathcal{E}(p_t)$, representation \eqref{eq:representation} applies.
		It is then enough to recognize that $V^H(m) \leqslant \frac{H-c}{1-\beta}$ and $V^L(m) \leqslant 0$ to conclude that $V(m \mid p_t, b_t) \leqslant 0$ for all $b_t \leqslant \Delta := \frac{(1-\beta)(c-L)}{H-\beta c - (1-\beta)L}$.
	\end{proof}

	\begin{lemma} \label{lem:vstar}
		Fix some equilibrium. Let $V^*(b_t) := V(b_t \mid p_t,b_t)$ denote $C_t$'s equilibrium continuation value of communicating her private belief $b_t$ truthfully. Then $V^*(b_t)$ is strictly increasing, (weakly) convex and differentiable for $b_t \geqslant \bar{p}$. Further, $V^*(b_t) = 0$ for $b_t < \bar{p}$ and $V^*(b_t) \geqslant 0$ for $b_t \geqslant \bar{p}$. Finally, for any $b_t,p_{t+1} \geqslant \bar{p}$, it is true that
		\begin{align} \label{eq:vconv}
			V(p_{t+1} \mid p_t, b_t) \leqslant V^*(b_t).
		\end{align}
	\end{lemma}
	\begin{proof}
		First note that $V^*(b_t)$ is well defined since $V(b_t \mid p_t,b_t)$ does not depend on $p_t$. Indeed, the path of play from $t+1$ onward only depends on $p_{t+1}$, and $C_t$'s expectation of this play only depends on $b_t$ and $p_{t+1}$. The claim that $V^*(b)=0$ for $b < \bar{p}$ is immediate from the fact that $C_{t+1}$ does not buy the product and receives utility zero when $p_{t+1} < \bar{p}$.	
		$V(p_{t+1} \mid p_t, b_t)$ can be expanded as
		\begin{multline*}
			V(p_{t+1} \mid p_t, b_t) = \theta(b_t) - c + \beta \cdot \\
			\sum\limits_{m \in \mathcal{M}} \left(b_t \int\limits_0^1 V(m | p_{t+1}, b) \cdot \mu(m|p_{t+1},b) d\Phi_H(b|p_{t+1}) + (1-b_t) \int\limits_0^1 V(m|p_{t+1},b) \cdot \mu(m|p_{t+1},b) d\Phi_L(b|p_{t+1})\right).
		\end{multline*}   
		At the same time $C_{t+1}$ with public belief $p_{t+1}$ when choosing the optimal communication strategy maximizes
		\begin{equation*}
		   \sum\limits_{m \in \mathcal{M}} \left(p_{t+1} \int\limits_0^1 V(m | p_{t+1}, b) \cdot \mu(m|p_{t+1},b) d\Phi_H(b|p_{t+1}) + (1-p_{t+1}) \int\limits_0^1 V(m|p_{t+1},b) \cdot \mu(m|p_{t+1},b) d\Phi_L(b|p_{t+1})\right).
		\end{equation*}
		Therefore if $p_{t+1} = b_t$, when $C_{t+1}$ maximizes the continuation value, he does it the same way as $C_t$ would, and therefore \eqref{eq:vconv} holds.
		
		To show that $V^*(b_t) > 0$ for $b_t > \bar{p}$ and $V^*(b_t) \geqslant 0$ for $b_t = \bar{p}$, note that by \eqref{eq:vconv},
		\begin{align*}
			V^*(b_t) \geqslant V(1 \mid p_t,b_t) =
			\sum\limits_{s=t+1}^{+\infty} \beta^{s-t-1} (\theta(b_t) - c) = \frac{1}{1-\beta} \cdot (\theta(b_t) - c).
		\end{align*}
		The expectation in the right hand side is then equal to zero at $b_t = \bar{p}$ and is strictly positive for $b_t > \bar{p}$.

		Next, we establish the differentiability of $V^*(b_t)$. The definition of $V^*(b_t)$ and \eqref{eq:vconv} imply that $V^*(b_t) = \max\limits_{p_{t+1}} V(p_{t+1} \mid p_t,b_t)$. Therefore by the envelope theorem and representation \eqref{eq:representation} we have 
		\begin{equation} \label{eq:V_derivative}
			\frac{d V^*(b_t)}{d b_t} = H - L + \beta \cdot \left(V^H(b_t) - V^L(b_t)\right).
		\end{equation}
		All terms in \eqref{eq:V_derivative} are well defined and bounded, therefore the derivative exists. The monotonicity then follows from \eqref{eq:V_derivative} because $H > L$ and $V^H(b_t) \geqslant 0 \geqslant V^L(b_t)$.
 
		Finally, we show the convexity of $V^*(b_t)$. Assume there are $\bar{p} \leqslant b' < b''$ and $b = \lambda \cdot b' + (1-\lambda) \cdot b''$ for some $\lambda \in (0,1)$. Lemma \ref{lem:representation} implies that $V(p_{t+1} \mid p_t,b_t)$ is linear and strictly increasing in $b_t$ for all $p_{t+1} \geqslant \bar{p}$, i.e., $V(p_{t+1} \mid p_t,b_t) = k_1(p_{t+1}) + k_2(p_{t+1}) \cdot b_t$ with $k_2(p_{t+1}) > 0$. Therefore
		\begin{equation*}
		\begin{aligned}
			V^*(b') \geqslant V(b | p_t, b') = k_1(b) + k_2 (b) \cdot b', \\
			V^*(b'') \geqslant V(b | p_t, b'') = k_1(b) + k_2(b) \cdot b''.
		\end{aligned}
		\end{equation*}
		Adding these two inequalities together with weights $\lambda$ and $1-\lambda$ respectively yields
		\begin{equation*}
			\lambda \cdot V^*(b') + (1-\lambda) \cdot V^*(b'') \geqslant k_1(b) + k_2(b) \cdot b = V^*(b),
		\end{equation*}
		which establishes the (weak) convexity of $V^*(b_t)$.
	\end{proof}

	Now we are ready to prove the main theorem.

	\subsection{Proof of Theorem \ref{thm:commitment}}
	\newcounter{pfstep}
	
	The proof of the Theorem proceeds in a number of steps. Steps $1$ to $10$ show that if an equilibrium exists, it must take the form described in the statement. Step $11$ then confirms that an equilibrium, in fact, exists and is unique up to payoff equivalence.

	\refstepcounter{pfstep}\subparagraph{Step \thepfstep.}
	\textbf{For any equilibrium there exists a payoff-equivalent equilibrium with the following property: for any $p_t \in [\bar{p},1)$ and any $b_t > \bar{p}$ we have $p_{t+1} \geqslant \bar{p}$.}
	
	If the property is satisfied for the original equilibrium, then we are done. Therefore, assume that given an equilibrium strategy $\mu$, there exist $p_t \in [\bar{p}, 1)$, $b_t > \bar{p}$, and $m < \bar{p}$ such that $\mu(m | p_t, b_t) > 0$. Consider then an alternative strategy $\mu'$, which coincides with $\mu$, with the exception that it communicates $b_t$ in state $p_t$ truthfully: $\mu'(b_t | p_t,b_t) = 1$.
	Doing so yields continuation value $V^*(b_t)$. This value is strictly positive for any $b_t > \bar{p}$ and can only be zero at $b_t = \bar{p}$ by Lemma \ref{lem:vstar}, while sending message $m < \bar{p}$ yields a continuation value of zero.
	At the same time, $q(m | p_t, \mu') \leqslant q(m | p_t, \mu)$ since $m$ is not sent after $b_t > \bar{p}$ under $\mu'$, hence message $m$ stops experimentation under $\mu'$ same as it does under $\mu$. Therefore, $\mu'$ yields a higher value than $\mu$ in private state $(p_t,b_t)$ and performs equally well in all other states. Exhausting all such $m$ we construct a payoff-equivalent equilibrium with a desired property.\footnote{Note that exhausting all such $p_t,b_t$, and $m$ cannot result in a strict improvement, since that would contradict $\mu$ being an equilibrium strategy.}
	
	\refstepcounter{pfstep}\subparagraph{Step \thepfstep.}
	\textbf{For any equilibrium there exists a payoff-equivalent equilibrium with the following property: for any $p_t \in [\bar{p},1)$ and any $b_t < \Delta := \frac{(1-\beta)(c-L)}{H-\beta c - (1-\beta)L}$ we have $p_{t+1} < \bar{p}$.}
	
	If the statement is true for the original equilibrium, we are done. If not, suppose that there exist such $p_t \in [\bar{p},1)$, some set of states $b_t$ contained in $[0, \Delta)$, and $m \geqslant \bar{p}$ such that $\mu(m | p_t, b_t) > 0$ according to the equilibrium strategy $\mu$. Let $\Phi(b|p_t)$ be defined by $\eqref{eq:phi}$. 
	Define $\gamma$ so that the following condition holds:
	\begin{equation}
		\label{eq:sameposterior}
		\frac{\int\limits_0^1 b \cdot \mu(m|p_t,b) \; d\Phi(b|p_t)}{\int\limits_0^1 \mu(m|p_t,b) \; d\Phi(b|p_t)} = \frac{\int\limits_\Delta^\gamma b \cdot \mu(m|p_t,b) \; d\Phi(b|p_t)}{\int\limits_\Delta^\gamma \mu(m|p_t,b) \; d\Phi(b|p_t)}.
	\end{equation}
	In words, we look at the equilibrium distribution of $b_t$ conditional on message $m$, and select $\gamma$ in such a way that the expectation of $b_t$ according to this distribution does not change if we restrict its domain from $[0,1]$ to $[\Delta, \gamma]$. Note that $\gamma < 1$ is well defined since $\mathbb{E}[b|p_t,m] = m > \bar{p} > \Delta$.
	
	Analogous to the argument in the previous step, we then construct an alternative strategy $\mu'$ that is identical to $\mu$, except whenever $\mu$ prescribes that message $m$ is sent, $\mu'$ prescribes the following:
	\begin{enumerate}
		\item if $b_t \in [0, \Delta)$, then send a truthful message $b_t$;
		\item for $b_t \in [\Delta, \gamma]$ send $m$ as in $\mu$;
		\item if $b_t \in (\gamma, 1]$, then send a truthful message $b_t$.
	\end{enumerate}
	Lemma \ref{lem:neg} implies that if $b_t < \Delta$, then $V(m \mid p_t, b_t) < 0$. Hence $\mu'$ performs strictly better than $\mu$ for $b_t < \Delta$ because $C_t$ stops experimentation with truthful message and therefore gets $0$. If $b_t \in [\Delta, \gamma]$ then $C_t$ receives the exact same payoff as above, since \eqref{eq:sameposterior} ensures that $m$ generates the same posterior under both $\mu$ and $\mu'$. 
	Finally, if $b_t \in (\gamma,1]$, then from Lemma \ref{lem:vstar} (and inequality \eqref{eq:vconv} in particular) we know that truthful reporting is weakly better for $C_t$ than inducing any other message. We conclude that strategy $\mu'$ weakly improves $C_t$'s continuation value after all private posteriors (and strictly at some) relative to $\mu$. Exhausting all such $m$ we construct a payoff-equivalent equilibrium with a desired property.
	
	\refstepcounter{pfstep}\subparagraph{Step \thepfstep.}
	\textbf{For any equilibrium there exists a payoff-equivalent equilibrium with the following property: for any $p_t \in [\bar{p},1)$ if $m \geqslant \bar{p}$ is such that $\mu(m|p_t,b'), \mu(m|p_t,b'') > 0$ for some $b', b'' \geqslant \bar{p}$, then $\mu(m|p_t,b) > 0$ for some $b < \bar{p}$.} In words, it is never optimal to pool posteriors above the cutoff $\bar{p}$ without also pooling them with some posteriors below $\bar{p}$.
 
	If the statement is true for the original equilibrium, we are done. If not, there exists some $m$ such that $\mu(m|p_t,b'), \mu(m|p_t,b'') > 0$ for some $b', b'' \geqslant \bar{p}$, and $\mu(m|p_t,b) = 0$ for all $b < \bar{p}$. Consider an alternative strategy $\mu'$ such that at all $(p_t,b_t)$, at which $\mu$ prescribes leaving review $m$, $\mu'$ prescribes instead revealing $b_t$ truthfully with the same probability, and $\mu'$ is equivalent to $\mu$ otherwise. At any such $b_t$, reporting $m$ yields continuation value $V(m|p_t,b_t)$, while a truthful report yields $V^*(b_t)$, which is weakly larger by \eqref{eq:vconv}, hence $\mu'$ is a weakly better strategy than $\mu$ for any $b_t$ for $C_t$. Adjusting communication strategies for all such $m$ we then construct a payoff-equivalent equilibrium with a desired property.
		
	\refstepcounter{pfstep}\subparagraph{Step \thepfstep.}
	\textbf{For any equilibrium there exists a payoff-equivalent equilibrium with the following property: for any $p_t \in [\bar{p},1)$, if there exists some pooling message $\bar{m}$ such that $\mu(\bar{m} | p_t,b'), \mu(\bar{m} | p_t,b'') > 0$ for some $b' < \bar{p} \leqslant b''$, then $\bar{m} = q(\bar{m} | p_t) = \bar{p}$.}
	If the statement is true for the original equilibrium, we are done. If not, there exists a message $\bar{m}$ with $q(\bar{m} | p_t) > \bar{p}$ which is sent for some $b' < \bar{p} \leqslant b''$. If according to $\Phi(b|p_t)$ the measure of states where $\bar{m}$ is sent is zero, then we are done as well, because we can substitute $\bar{m}$ with truthful reports in these states. Therefore, assume that the measure of states where $\bar{m}$ is sent is non-zero.
	Consider then an alternative strategy $\mu'$ which is equivalent to $\mu$, except that message $\bar{m}$ is replaced by the following messages. Whenever $\mu$ prescribes message $\bar{m}$, $\mu'$ sends message $\bar{p}$ if $b_t \leqslant \gamma$, and reports $b_t$ truthfully whenever $b_t > \gamma$ (with the same probability as under $\mu$). Select $\gamma$ in such a way that $q(\bar{p}| p_t, \mu') = q(\bar{m}| p_t, \mu, b_t < \gamma) = \bar{p}$.\footnote{Note that such $\gamma$ always exists because we assumed that $\bar{m}$ is sent for a non-zero measure of states $b_t$.}
	
	Consider then $C_t$'s expected continuation value under $\mu$ and under $\mu'$. Because these strategies are identical, except for private beliefs where $\bar{m}$ is communicated under $\mu$, we need to compare the respective contributions to the expected value of $C_t$. For strategy $\mu$ it is
	\begin{equation*}
		\int \limits_{b : \mu(\bar{m} | p_t, b) > 0} V(\bar{m}|p_t,b) \cdot \mu(\bar{m} | p_t, b) \; d \Phi(b | p_t) = V^*(\bar{m}) \cdot \int \limits_{b : \mu(\bar{m}|p_t,b) > 0} \mu(\bar{m}|p_t,b) \; d \Phi(b|p_t)
	\end{equation*}
	In the above equality we used the linearity of $V(m|p_t,b_t)$ in $b_t$ established in Lemma \ref{lem:representation}. For $\mu'$ the contribution is
	\begin{align*}
		& \int \limits_{b \leqslant \gamma : \mu(\bar{m}|p_t,b) > 0} V(\bar{p} | p_t, b) \cdot \mu(\bar{m}|p_t,b) \; d \Phi(b|p_t) + \int \limits_{b > \gamma : \mu(\bar{m} | p_t, b) > 0} V(b | p_t, b) \cdot \mu(\bar{m}|p_t,b) \; d \Phi(b|p_t) = \\
		& V^*(\bar{p}) \cdot \int \limits_{b \leqslant \gamma : \mu(\bar{m} | p_t, b) > 0} \mu(\bar{m}|p_t,b) \; d \Phi(b|p_t) + \int \limits_{b > \gamma : \mu(\bar{m} | p_t, b) > 0} V^*(b) \cdot \mu(\bar{m}|p_t,b) \; d \Phi(b|p_t)
	\end{align*}
	Since $V^*(b_t)$ is (weakly) convex, Jensen's inequality implies that strategy $\mu'$ is (weakly) better than $\mu$. Taking all such $\bar{m}$ and adjusting the communication strategy we then construct a payoff-equivalent equilibrium with a desired property.

	\refstepcounter{pfstep}\subparagraph{Step \thepfstep.}
	\textbf{It is without loss to restrict attention to Markov equilibria with at most one pooling message which induces public belief $p_{t+1}=\bar{p}$.} This is because by Step $4$, any pooling message must induce the same public posterior $p_{t+1} = \bar{p}$, and by the Markov property, the continuation play must then also be the same after any pooling message.

	\medskip
	To summarize the steps above: we now know that any Markov equilibrium is payoff-equivalent to a one where $C_t$ either reveals her belief truthfully so that $p_{t+1} = b_t$, or sends a pooling message $m = \bar{p}$ so that $p_{t+1} = \bar{p}$.
	\medskip
	
	\refstepcounter{pfstep}\subparagraph{Step \thepfstep.}
	\textbf{For any equilibrium there exists a payoff-equivalent equilibrium with the following property: for any $p_t \in [\bar{p},1)$, there exists $l(p_t) \leqslant \bar{p}$ such that $\mu(\bar{p} | p_t, b) = 0$ for $b < l(p_t)$, and $\mu(\bar{p} | p_t, b) = 1$ for $b \in [l(p_t), \bar{p}]$.} In other words, the pooling region (if it is nonempty) is convex left of the cutoff.
	
	For the original equilibrium, denote $l(p_t) := \inf \{ b \leqslant \bar{p} \; | \; \mu(\bar{p} | p_t, b) > 0\}$. If the property holds for the original equilibrium with such $l(p_t)$ then we are done. If not, then there exists a subset of states $b_t$ contained in $[l(p_t), \bar{p}]$ where $b_t$ is revealed truthfully (remember from the previous step that induced posterior is either truthful, or is equal to $\bar{p}$). If the total measure of such states is zero, then we are done as well. Indeed, we can prescribe sending message $\bar{p}$ in these states: adjusted communication strategy satisfies the desired property while the $C_t$'s value does not change. Therefore, assume that the total measure of such states is strictly positive, i.e.,
	\begin{equation*}
		\int\limits_{l(p_t)}^{\bar{p}} \left(1 - \mu (\bar{p} | p_t, b)\right) \; d\Phi(b|p_t) > 0.
	\end{equation*}
	We next show that it implies a contradiction with $\mu$ being an equilibrium strategy. For that we construct an alternative strategy $\mu'$ which delivers a strictly higher payoff for $C_t$ rather than $\mu$. For that, define $l'(p_t) > l(p_t)$ such that
	\begin{equation*}
		\int\limits_{l'(p_t)}^{\bar{p}} \; d\Phi(b|p_t) = \int\limits_{l(p_t)}^{\bar{p}} \mu (\bar{p} | p_t, b) \; d\Phi(b|p_t).
	\end{equation*}
	Also define $\gamma \geqslant \bar{p}$ such that
	\begin{equation} \label{eq:mu_prime}
		\mathbb{E} \left[ b \; | \; b \in [l'(p_t), \bar{p}] \cup S_\gamma \right] = \bar{p},
	\end{equation}
	where $S_\gamma = \{b \; | \; \mu(\bar{p} | p_t, b) > 0, \bar{p} < b \leqslant \gamma\}$. In words, below $\bar{p}$, strategy $\mu'$ compresses the entire mass of beliefs where message $\bar{p}$ is transmitted towards $\bar{p}$. Above $\bar{p}$, it cuts high beliefs to maintain the induced posterior after the pooling message to be exactly $\bar{p}$.
	
	Define $\mu'$: it sends pooling message $\bar{p}$ for all $b_t \in [l'(p_t), \bar{p}] \cup S_\gamma$ and reveals $b_t$ truthfully otherwise.\footnote{Pooling such private beliefs into one message indeed induces posterior $\bar{p}$ due to \eqref{eq:mu_prime}.} Then we need to show
	\begin{multline*}
		\int\limits_{l'(p_t)}^{\bar{p}} V(\bar{p} | p_t, b) \; d\Phi(b|p_t) + \int\limits_{\bar{p}}^{\gamma} \left[V(\bar{p} | p_t, b) \cdot \mu (\bar{p} | p_t, b) + V^*(b) \cdot (1 - \mu (\bar{p} | p_t, b))\right] \; d\Phi(b|p_t) +
		\int\limits_{\gamma}^{1} V^*(b) \; d\Phi(b|p_t)
		>  \\
		\int\limits_{l(p_t)}^{\bar{p}} V(\bar{p} | p_t, b) \cdot \mu (\bar{p} | p_t, b) \; d\Phi(b|p_t) + \int\limits_{\bar{p}}^{1} \left[V(\bar{p} | p_t, b) \cdot \mu (\bar{p} | p_t, b) + V^*(b) \cdot (1 - \mu (\bar{p} | p_t, b))\right] \; d\Phi(b|p_t).
	\end{multline*}
	Because $V(\bar{p} | p_t, b_t)$ is linear in $b_t$ and due to \eqref{eq:mu_prime} we have
	\begin{multline*}
		\int\limits_{l'(p_t)}^{\bar{p}} V(\bar{p} | p_t, b) \; d\Phi(b|p_t) + \int\limits_{\bar{p}}^{\gamma} V(\bar{p} | p_t, b) \cdot \mu (\bar{p} | p_t, b) \; d\Phi(b|p_t) = \\ \int\limits_{l(p_t)}^{\bar{p}} V(\bar{p} | p_t, b) \cdot \mu (\bar{p} | p_t, b) \; d\Phi(b|p_t) + \int\limits_{\bar{p}}^{1} V(\bar{p} | p_t, b) \cdot \mu (\bar{p} | p_t, b) \; d\Phi(b|p_t)
	\end{multline*}
	Therefore it only remains to show that
	\begin{equation*}
		\int\limits_{\gamma}^{1} V^*(b) \cdot \mu (\bar{p} | p_t,b) d\Phi(b|p_t) > 0.
	\end{equation*}
	It is true because the measure of beliefs $b_t > \gamma$ for which $\bar{p}$ is transmitted is strictly positive, and $V^*(\gamma) > V^*(\bar{p}) \geqslant 0$. This shows that $\mu'$ gives a strictly higher payoff to $C_t$ rather than $\mu$, which is a contradiction.

	\refstepcounter{pfstep}\subparagraph{Step \thepfstep.}
	\textbf{For any equilibrium there exists a payoff-equivalent equilibrium with the following property: for any $p_t \in [\bar{p},1)$, there exists $r(p_t) \geqslant \bar{p}$ such that $\mu(\bar{p} | p_t, b) = 0$ for $b > r(p_t)$, and $\mu(\bar{p} | p_t, b) = 1$ for $b \in [\bar{p}, r(p_t)]$.} In other words, the pooling region (if it is nonempty) is convex right of the cutoff.
	
	For the original equilibrium define $r(p_t) := \sup \{ b \geqslant \bar{p} \; | \; \mu(\bar{p} | p_t, b) > 0\}$. If the property holds for the original equilibrium with such $r(p_t)$ then we are done. If not, then there exists a subset of states $b_t$ contained in $[\bar{p}, r(p_t)]$ where $b_t$ is revealed truthfully. If the total measure of such states is zero, then we are done as well.\footnote{For the same reason as in the previous step.} Therefore, assume that this measure is strictly positive, i.e.,
	\begin{equation*}
		\int\limits_{\bar{p}}^{r(p_t)} \left(1 - \mu (\bar{p} | p_t, b)\right) \; d\Phi(b|p_t) > 0.
	\end{equation*}
	We next construct an alternative strategy $\mu'$ which delivers a (weakly) higher payoff for $C_t$ rather than $\mu$ and satisfies the desired property. From the previous step we know there exists $l(p_t) \leqslant \bar{p}$ such that $b_t$ is revealed truthfully for $b_t < l(p_t)$ and message $\bar{p}$ is sent for all $b_t \in [l(p_t),\bar{p}]$. Then define $r'(p_t) < r(p_t)$ such that
	\begin{equation} \label{eq:r_def}
		\mathbb{E} \left[ b \; | \; b \in [l(p_t), r'(p_t)] \right] = \mathbb{E} \left[ b \; | \; \mu(\bar{p} | p_t, b) > 0 \right] = \bar{p},
	\end{equation}
	Consider an alternative strategy $\mu'$ for $C_t$ that is equivalent to $\mu$ for $b_t < \bar{p}$, sends message $\bar{p}$ for all $b_t \in [\bar{p}, r'(p_t)]$, and reveals $b_t$ truthfully for all $b_t > r'(p_t)$. We next show that $\mu'$ yields a (weakly) higher payoff to $C_t$. The equivalence of $\mu$ and $\mu'$ for $b_t < \bar{p}$ implies that we only need to look at $b_t \geqslant \bar{p}$. Hence, we need to show	
	\begin{multline*}
		\int\limits_{\bar{p}}^1 \Big[ V(\bar{p} \mid p_t,b) \cdot \mu(\bar{p} \mid p_t, b) + V^*(b) \cdot \left(1 - \mu(\bar{p} \mid p_t, b) \right) \Big] \; d\Phi(b | p_t) \leqslant \\
		\int\limits_{\bar{p}}^{r'(p_t)} V(\bar{p} \mid p_t,b) \; d\Phi(b | p_t) + \int\limits_{r'(p_t)}^1 V^*(b) \; d\Phi(b | p_t)
	\end{multline*}
	Rearranging terms we get
	\begin{multline} \label{eq:t1_valueineq4}
		\int\limits_{r'(p_t)}^1 \Big[ V^*(b) - V(\bar{p} \mid p_t,b) \Big] \cdot \mu(\bar{p} \mid p_t, b) \; d\Phi(b|p_t) \geqslant \\
		\int\limits_{\bar{p}}^{r'(p_t)} \Big[ V^*(b) - V(\bar{p} \mid p_t,b) \Big] \cdot \left( 1- \mu(\bar{p} \mid p_t, b) \right) \; d\Phi(b|p_t).
	\end{multline}
	First, expand definition \eqref{eq:r_def}. We have
	\begin{equation*}
		\int\limits_{l(p_t)}^{r'(p_t)} (b - \bar{p}) \; d\Phi(b | p_t) = \int\limits_{l(p_t)}^{\bar{p}} (b - \bar{p}) \; d\Phi(b | p_t) + \int\limits_{\bar{p}}^{1} (b - \bar{p}) \cdot \mu(\bar{p}|p_t,b) \; d\Phi(b | p_t)
	\end{equation*}
	Canceling terms on both sides we get
	\begin{equation} \label{eq:t1_beliefeq1}
		\int\limits_{\bar{p}}^{r'(p_t)} (b - \bar{p}) \cdot (1-\mu(\bar{p}|p_t,b)) \; d\Phi(b | p_t) = \int\limits_{r'(p_t)}^{1} (b - \bar{p}) \cdot \mu(\bar{p}|p_t,b) \; d\Phi(b | p_t)
	\end{equation}
	Second, since $V(\bar{p} \mid p_t,b_t)$ is linear in $b_t$ by Lemma \ref{lem:representation}, and $V^*(b_t)$ is convex by Lemma \ref{lem:vstar}, their difference $V^*(b_t) - V(\bar{p} \mid p_t,b_t)$ is convex in $b_t$. Hence by Lemma \ref{lem:mon_con}, $\frac{V^*(b_t) - V(\bar{p} \mid p_t,b_t)}{b_t-\bar{p}}$ is increasing for $b_t \geqslant \bar{p}$. This monotonicity implies
	\begin{multline*}
		\int\limits_{r'(p_t)}^1 (b-\bar{p}) \cdot \frac{V^*(r(p_t)) - V(\bar{p} \mid p_t,r(p_t))}{r(p_t)-\bar{p}} \cdot \mu(\bar{p} \mid p_t, b) \; d\Phi(b|p_t) \leqslant
		\\ \int\limits_{r'(p_t)}^1 (b-\bar{p}) \cdot \frac{V^*(b) - V(\bar{p} \mid p_t,b)}{b-\bar{p}} \cdot \mu(\bar{p} \mid p_t, b) \; d\Phi(b|p_t) =
		\\ \int\limits_{r'(p_t)}^1 \Big[ V^*(b) - V(\bar{p} \mid p_t,b) \Big] \cdot \mu(\bar{p} \mid p_t, b) \; d\Phi(b|p_t),
	\end{multline*}
	where the top line is the LHS of \eqref{eq:t1_beliefeq1} multiplied by a constant $\frac{V^*(r(p_t)) - V(\bar{p} \mid p_t,r(p_t))}{r(p_t)-\bar{p}}$, and the bottom line is the LHS of \eqref{eq:t1_valueineq4}. For the RHS of \eqref{eq:t1_beliefeq1} and \eqref{eq:t1_valueineq4} we similarly get
	\begin{multline*}
		\int\limits_{\bar{p}}^{r'(p_t)} (b-\bar{p}) \cdot \frac{V^*(r(p_t)) - V(\bar{p} \mid p_t,r(p_t))}{r(p_t)-\bar{p}} \cdot \left(1 - \mu(\bar{p} \mid p_t, b) \right) \; d\Phi(b|p_t) \geqslant
		\\ \int\limits_{\bar{p}}^{r'(p_t)} (b-\bar{p}) \cdot \frac{V^*(b) - V(\bar{p} \mid p_t,b)}{b-\bar{p}} \cdot \left(1 - \mu(\bar{p} \mid p_t, b) \right) \; d\Phi(b|p_t) =
		\\ \int\limits_{\bar{p}}^{r'(p_t)} \Big[ V^*(b) - V(\bar{p} \mid p_t,b) \Big] \cdot \left(1 - \mu(\bar{p} \mid p_t, b) \right) \; d\Phi(b|p_t).
	\end{multline*}
	The inequalities above together with \eqref{eq:t1_beliefeq1} imply that \eqref{eq:t1_valueineq4} holds, meaning that $\mu'$ satisfies the desired property and is a weakly better strategy for $C_t$ rather than $\mu$, and is thus payoff-equivalent to $\mu$.

	\refstepcounter{pfstep}\subparagraph{Step \thepfstep.}
	\textbf{In equilibrium, $V^*(\bar{p}) > 0$.}
 
	The previous steps establish that any equilibrium is payoff equivalent to one in which $C_t$ sends a pooling message $\bar{p}$ for $b_t \in \left[l(p_t), r(p_t)\right]$ and reveals $b_t$ truthfully otherwise. Therefore
	\begin{align*}
		V^*(\bar{p}) = V(\bar{p} \mid p_t, \bar{p}) = \theta(\bar{p}) - c + \beta \cdot \int\limits_{l(\bar{p})}^{r(\bar{p})} V^*(\bar{p}) \; d\Phi(b|\bar{p}) + \beta \cdot \int\limits_{r(\bar{p})}^1 V^*(b) \; d\Phi(b|\bar{p}).
	\end{align*}
	By definition of $\bar{p}$, the first summand is $0$. Note that $r(\bar{p}) < 1$ because $l(\bar{p}) \geqslant \Delta$. Therefore, the last term is strictly positive by the properties of $V^*(b)$ outlined in Lemma \ref{lem:vstar}, and the fact that $b_{t+1}$ has full support on $[0,1]$.
 
	\refstepcounter{pfstep}\subparagraph{Step \thepfstep.}
	\textbf{Assume an equilibrium satisfies properties from Steps 1-7, then for any $p_t \in [\bar{p}, 1)$ the pooling interval is nonempty: $l(p_t) < \bar{p} < r(p_t)$.}
 
	Suppose the contrary -- which, given the steps above, means that the equilibrium communication strategy $\mu$ is truthful: $\mu(b_t \mid p_t, b_t) =  1$ for all $b_t$. Take some $\varepsilon \in (0, \bar{p})$. Denote $l(\varepsilon) := \bar{p} - \varepsilon$ and set $r(\varepsilon)$ such that $\mathbb{E}[b_t \mid b_t \in [l(\varepsilon), r(\varepsilon)]] = \bar{p}$, which is possible by Lemma \ref{lem:eps}. 
	Consider then a strategy $\mu_\varepsilon$ that sends a pooling message $\bar{p}$ for all $b_t \in [l(\varepsilon), r(\varepsilon)]$ and is truthful otherwise. We now show that there exists $\varepsilon$ such that $\mu_\varepsilon$ yields a higher expected payoff for $C_t$ in state $p_t$ rather than $\mu$.
	
	The difference in $C_t$'s expected payoff under $\mu_\varepsilon$ relative to that under $\mu$ is given by
	\begin{align*}
		\int\limits_{l(\varepsilon)}^{r(\varepsilon)} V(\bar{p} \mid p_t, b) \; d\Phi(b | p_t) - \int\limits_{\bar{p}}^{r(\varepsilon)} V^*(b) \; d\Phi(b | p_t).
	\end{align*}
	Our goal is to show that for some $\varepsilon$, the expression above is strictly positive, i.e.,
	\begin{align} \label{eq:t1_valueineq5}
		\int\limits_{l(\varepsilon)}^{\bar{p}} V(\bar{p} \mid p_t, b) \; d\Phi(b | p_t)
		&>
		\int\limits_{\bar{p}}^{r(\varepsilon)} \left[ V^*(b) - V(\bar{p} \mid p_t, b) \right] \; d\Phi(b | p_t).
	\end{align}
	Since $V(m \mid p_t, b_t)$ is linear in $b_t$:
	\begin{align*}
		\int \limits_{l(\varepsilon)}^{\bar{p}} \left[ V^*(\bar{p}) - V(\bar{p} \mid p_t, b) \right] \; d\Phi (b|p_t) &= \int \limits_{\bar{p}}^{r(\varepsilon)} \left[ V(\bar{p} \mid p_t, b) - V^*(\bar{p}) \right] \; d\Phi (b|p_t).
	\end{align*}
	Adding the equality above to \eqref{eq:t1_valueineq5} we get
	\begin{equation*} \label{eq:t1_valueineq6}
		\int\limits_{l(\varepsilon)}^{\bar{p}} V^*(\bar{p}) \; d\Phi(b|p_t) > \int\limits_{\bar{p}}^{r(\varepsilon)} \left[ V^*(b) - V^*(\bar{p}) \right] \; d\Phi(b | p_t).
		\end{equation*}
	Monotonicity of $V^*(b_t)$ implies it is sufficient to show: 
	\begin{align}  \label{eq:t1_valueineq7}
		V^*(\bar{p}) \cdot \left[ \Phi(\bar{p}) - \Phi(l(\varepsilon)) \right] 
		&> 
		\left[ V^*(r(\varepsilon)) - V^*(\bar{p}) \right] \cdot \left[ \Phi(r(\varepsilon)) - \Phi(\bar{p}) \right],
	\end{align}
	Recall from Lemma \ref{lem:eps} that there exist $B_l := \min \{ \phi(b|p_t) \mid b \in [l(\varepsilon), r(\varepsilon)] \}$ and $B_h := \max \{ \phi(b|p_t) \mid b \in [l(\varepsilon), r(\varepsilon)] \}$, -- bounds on density $\phi(b|p_t)$ on the relevant interval of $b$. Since $\phi (b|p_t)$ is continuous and strictly positive, these bounds exist and $0 < B_l < B_h < \infty$. Given these bounds, it was established that $r(\varepsilon) - \bar{p} \leqslant \sqrt{\frac{B_h}{B_l}} \cdot \varepsilon$. These bounds then imply that 
	\begin{align} \label{eq:Phibound}
		\Phi(\bar{p}) - \Phi(l(\varepsilon)) & \geqslant B_l \cdot \varepsilon,
		&
		\Phi(r(\varepsilon)) - \Phi(\bar{p}) & \leqslant B_h \sqrt{\frac{B_h}{B_l}} \cdot \varepsilon.
	\end{align}
	On the other hand, $V^*(r(\varepsilon)) - V^*(\bar{p})$ is also bounded from above by a factor of $\varepsilon$. 
	To see this, recall that Lemma \ref{lem:mon_con} provided an exact form for $\frac{dV^*(b_t)}{db_t}$. Using the trivial bounds for $V^H(b_t)$ and $V^L(b_t)$ we conclude that $\frac{dV^*(b_t)}{db_t} \leqslant \frac{H - (1-\beta) \cdot L - \beta \cdot c}{1 - \beta}$. Therefore
	\begin{equation} \label{eq:t1_valueineq9}
		V^*(r(\varepsilon)) - V^*(\bar{p}) \leqslant \frac{H - (1-\beta) \cdot L - \beta \cdot c}{1 - \beta} \cdot \left(r(\varepsilon) - \bar{p}\right) \leqslant \frac{H - (1-\beta) \cdot L - \beta \cdot c}{1 - \beta} \cdot \sqrt{\frac{B_h}{B_l}} \cdot \varepsilon.
	\end{equation}
	Combining \eqref{eq:t1_valueineq9} with \eqref{eq:Phibound}, we have that \eqref{eq:t1_valueineq7} holds if the following holds:
	\begin{align} \label{eq:t1_valueineq8}
		V^*(\bar{p}) \cdot B_l \cdot \varepsilon > \frac{H - (1-\beta) \cdot L - \beta \cdot c}{1 - \beta} \cdot \frac{B_h^3}{B_l^2} \cdot \varepsilon^2.
	\end{align}
	If \ref{eq:t1_valueineq8} holds for the picked $\varepsilon$ we are done. If not, since the LHS is proportional to $\varepsilon$ and the RHS to $\varepsilon^2$, there exists a small enough $\tilde{\varepsilon} < \varepsilon$ such that \eqref{eq:t1_valueineq8} holds. For $\tilde{\varepsilon} < \varepsilon$ both inequalities in \ref{eq:Phibound} are valid and therefore \eqref{eq:t1_valueineq5} holds for $\tilde{\varepsilon}$ as well, meaning that strategy $\mu_{\tilde{\varepsilon}}$ is strictly better for $C_t$ rather than $\mu$.

	\bigskip
	This concludes the proof of the shape of equilibrium. The two remaining steps verify that an equilibrium of this form does, indeed, exist, and is unique.

	\refstepcounter{pfstep}\subparagraph{Step \thepfstep.}
	\textbf{Assume an equilibrium satisfies properties from Steps 1-7, then for any $p_t \in [\bar{p}, 1)$ the optimal pooling interval $[l(p_t), r(p_t)]$ is uniquely determined.}

	For a given $l(p_t)$, the right end $r(p_t)$ of the pooling interval is uniquely determined by Lemma \ref{lem:eps}. Therefore, it is sufficient to show that $l(p_t)$ is uniquely determined for any $p_t \in [\bar{p},1)$.
	Applying the implicit function theorem to \eqref{eq:rofl2} we obtain
	\begin{equation} \label{eq:implicit}
		\frac{dr(l)}{dl} = -\frac{\bar{p}-l}{r(l)-\bar{p}} \cdot \frac{\phi(l|p_t)}{\phi(r(l)|p_t)}.
	\end{equation}
	Given \eqref{eq:implicit}, $C_t$ chooses an optimal $l$ to maximize
	\begin{equation*}
		V^*(\bar{p}) \cdot \int\limits_{l}^{r} \; d\Phi(b|p_t) + \int\limits_{r}^{1} V^*(b) \; d\Phi(b|p_t).
	\end{equation*}
	Differentiating the above expression with respect to $l$ we get
	\begin{equation*}
		V^*(\bar{p}) \left(\phi(r | p_t) \cdot \frac{dr(l)}{dl} - \phi(l | p_t) \right) - V^*(r) \cdot \phi(r | p_t) \cdot \frac{dr(l)}{dl}.
	\end{equation*}
	Given \eqref{eq:implicit} the above expression can be rewritten as
	\begin{equation} \label{eq:derivative}
		\phi(l | p_t) \cdot \left(\frac{V^*(r(l)) - V^*(\bar{p})}{r(l) - \bar{p}} \cdot (\bar{p} - l) - V^*(\bar{p}) \right).
	\end{equation}
	We next show that there exists $l^*$ such that \eqref{eq:derivative} is strictly positive for $l < l^*$, equal to zero at $l^*$ and is strictly negative for $l > l^*$. by Lemma \ref{lem:mon_con}, $\frac{V^*(r(l)) - V^*(\bar{p})}{r(l) - \bar{p}}$ is an increasing function of $r$, and $r(l)$ is a decreasing function of $l$ by \eqref{eq:rofl2}. Therefore, the expression in brackets is a decreasing function of $l$. What remains is to show that \eqref{eq:derivative} is positive for $l = 0$ and is negative for $l = \bar{p}$. The latter is trivial, while the former reduces to ${V^*}'(\bar{p}) \cdot \bar{p} - V^*(\bar{p}) > 0$. Given \eqref{eq:V_derivative} this expression can be calculated explicitly:
	\begin{equation*}
		{V^*}'(\bar{p}) \cdot \bar{p} - V^*(\bar{p}) = \bar{p} \cdot (H - L) + \beta \cdot \bar{p} \cdot \left(V^H(\bar{p}) - V^L(\bar{p})\right) - \beta \cdot \bar{p} \cdot V^H(\bar{p}) - \beta \cdot (1-\bar{p}) \cdot V^L(\bar{p}) =  \bar{p} \cdot (H - L) - \beta \cdot V^L(\bar{p}) > 0
	\end{equation*}This implies $l^*$ is the optimal left end of the pooling interval, and $r^*$ derived from \eqref{eq:rofl2} is then the optimal right end of the pooling interval.

	\refstepcounter{pfstep}\subparagraph{Step \thepfstep.}
	\textbf{An equilibrium exists and is unique up to payoff equivalence.}
	
	Remember that a Markov equilibrium is a collection of communication strategies defined for all $p_t$. Steps 1-10 established that to describe the optimal strategy for any $p_t \in [\bar{p},1)$ it is enough to determine the pooling region $[l(p_t), r(p_t)]$ such that the bounds of the interval satisfy \eqref{eq:rofl2}.
	
	Notice that $V^*(p_t)$ can be interpreted as the value of consumer $C_t$ with public belief $p_t$, evaluated \textit{before} $C_t$ purchases the product. Therefore, $V^*(p_t)$ satisfies the following Bellman equation for any $p_t \in [\bar{p}, 1)$:
	\begin{align} \label{eq:Bellman}
		V^*(p_t) = \max\limits_{l,r} & \left[\theta(p_t) - c + \beta \cdot \int\limits_{l}^{r} V^*(\bar{p}) \; d\Phi(b | p_t) + \beta \cdot \int\limits_{r}^{1} V^*(b) \; d\Phi(b | p_t)\right] \\
		\nonumber
		\text{s.t.} & \int \limits_{\bar{p}}^{r} (b-\bar{p}) \cdot \phi (b|p_t) \; db = \int \limits_{l}^{\bar{p}} (\bar{p}-b) \cdot \phi (b|p_t) \; db.
	\end{align}
	We next apply the Banach fixed-point theorem to establish that there exists a unique $V^*(p_t)$ which constitutes a solution to \eqref{eq:Bellman}. To do this, consider a metric space of functions $\mathcal{C^I}[\bar{p}, 1]$ defined on $[\bar{p}, 1]$ such that $f \in \mathcal{C^I}[\bar{p}, 1]$ if:
	\begin{enumerate}
		\item $f$ is continuous on $[\bar{p}, 1]$,
		\item $f(\bar{p}) \geqslant 0$, $f(1) = \frac{H-c}{1-\beta}$,
		\item $f$ is (weakly) increasing on $[\bar{p},1]$. 
	\end{enumerate}
	We endow it with a standard metric $\rho(f,g) = \max\limits_{x \in [\bar{p}, 1]} |f(x) - g(x)|$.\footnote{Note that this is equivalent to the standard sup-norm $||\cdot||_\infty$. Since the functions are continuous, the maximum exists.} Define mapping $\mathcal{C^I}[\bar{p}, 1] \rightarrow \mathcal{C^I}[\bar{p}, 1]$ as
	\begin{align} \label{eq:mapping}
		T_f(p_t) = \max\limits_{l, r} & \left[\theta(p_t) - c + \beta \cdot \int\limits_{l}^{r} f(\bar{p}) \; d\Phi(b | p_t) + \beta \cdot \int\limits_{r}^{1} f(b) \; d\Phi(b | p_t)\right]
		\\
		\nonumber
		\text{s.t.} & \int \limits_{\bar{p}}^{r} (b-\bar{p}) \cdot \phi (b|p_t) \; db = \int \limits_{l}^{\bar{p}} (\bar{p}-b) \cdot \phi (b|p_t) \; db.
	\end{align}
	We separate this step into three sub-steps. These sub-steps verify that the premise of the Banach fixed-point theorem is satisfied.

	\subparagraph{Step \thepfstep.1}
	\textbf{$T_f(p_t)$ is well defined, i.e., $T_f(p_t)$ maps $\mathcal{C^I}[\bar{p}, 1]$ into itself.}
	
	To begin with, note that the solution to the maximization problem in the RHS of \eqref{eq:mapping} exists, and thus $T_f(p_t)$ is well defined. Indeed, Lemma \ref{lem:eps} shows that $r(l,p_t)$ is a continuous function. Therefore, the expression within the brackets in the RHS of \eqref{eq:mapping} is a continuous function of $l$, and thus it attains its maximum in $l$ on $[0, \bar{p}]$.
	We now proceed to verify that $T_f$ satisfies the properties required by $\mathcal{C^I}[\bar{p}, 1]$.
	
	First, let us show that $T_f(\bar{p}) \geqslant 0$ and $T_f(1) = \frac{H-c}{1-\beta}$. If $f(\bar{p}) \geqslant 0$ and $f$ is weakly increasing, then $f(x) \geqslant 0$ for all $x \in [\bar{p},1]$. Therefore $T_f(\bar{p}) \geqslant 0$ no matter how $l,r$ are chosen in \eqref{eq:mapping}. With $p_t = 1$, distribution $\Phi(b|1)$ is degenerate with a point mass in $b = 1$. Therefore, \eqref{eq:mapping} does not depend on $l,r$, and we have
	\begin{equation*}
		T_f(1) = H - c + \beta \cdot \frac{H-c}{1-\beta} = \frac{H-c}{1-\beta}.
	\end{equation*}
	
	Second, we establish the monotonicity of $T_f$. Fix any $x_1 < x_2$. Let $l_1, r_1$ and $l_2, r_2$ be the respective maximizers of \eqref{eq:mapping} at $p_t = x_1,x_2$. Also define $r'_1 := r(l_1, x_2)$, i.e., the value of $r$ that solves \eqref{eq:rofl2} with $l = l_1$ and $p_t = x_2$. Then
	\small
	\begin{align*}
		T_f(x_1) &= \theta(x_1) - c + \beta \cdot \left(\int\limits_{l_1}^{r_1} f(\bar{p}) \; d\Phi(b | x_1) + \int\limits_{r_1}^{1} f(b) \; d\Phi(b | x_1)\right) \\
		&< \theta(x_2) - c + \beta \left(\int\limits_{l_1}^{r_1} f(\bar{p}) \; d\Phi(b | x_2) + \int\limits_{r_1}^{1} f(b) \; d\Phi(b | x_2) \right) \\
		& \leqslant \theta(x_2) - c + \beta\left(\int\limits_{l_1}^{r'_1} f(\bar{p}) \; d\Phi(b | x_2) + \int\limits_{r'_1}^{1} f(b) \; d\Phi(b | x_2) \right) \\
		& \leqslant \theta(x_2) - c + \beta\left(\int\limits_{l_2}^{r_2} f(\bar{p}) \; d\Phi(b | x_2) + \int\limits_{r_2}^{1} f(b) \; d\Phi(b | x_2) \right) = T_f(x_2).
	\end{align*}
	\normalsize
	The first inequality in the sequence above is valid because $\Phi(b|x_2)$ first-order stochastically dominates $\Phi(b|x_1)$ and $f(b)$ is increasing. The second inequality follows from $r'_1 < r_1$ and $f(b)$ being increasing. The last inequality holds because $l_1, r'_1$ are not necessarily the maximizers for \eqref{eq:mapping} given $p_t = x_2$.
	
	Finally, we establish the continuity of $T_f(p_t)$ in $p_t$. We first show that $T_f(p_t)$ is continuous at $p_t = 1$. $T_f(1) = \frac{H-c}{1-\beta}$. Take any $\varepsilon > 0$. Because $f(b)$ is continuous and strictly increasing, there exists $\Delta > 0$ such that for all $b \in [1-\Delta, 1]$ we have $f(b) \geqslant \frac{H-c}{1-\beta} - \frac{\varepsilon}{3}$. Then take $\delta_1$ such that $\Phi(1-\Delta | 1- \delta_1) < \frac{\varepsilon}{3 \cdot \frac{H-c}{1-\beta}}$ and $\delta_2 = \frac{\varepsilon}{3 (H-L)}$. Then for any $\delta < \min\{\delta_1, \delta_2\}$:
	\begin{equation*}
		T_f(1-\delta) > H-c - \frac{\varepsilon}{3} + 0 + \beta \cdot \left(1-\frac{\varepsilon}{3 \cdot \frac{H-c}{1-\beta}}\right) \cdot \left(\frac{H-c}{1-\beta} - \frac{\varepsilon}{3}\right) > \frac{H-c}{1-\beta} - \varepsilon,
	\end{equation*}
	which implies that $T_f(p_t)$ is continuous at $p_t = 1$. Next, we show that $T_f(p_t)$ is continuous for $p_t \in [\bar{p}, 1-\delta]$ for any given $\delta > 0$. We can apply Berge's Maximum Theorem (see \citet*{SLP}, Theorem 3.6). Consider \eqref{eq:Bellman} as a maximization problem in $l$ and treat $r = r(l, p_t)$ that solves the constraint as a function of $l$ and $p_t$. Then for any $p_t \in [\bar{p}, 1-\delta]$, the set of admissible $l$ is a closed interval $[0,\bar{p}]$. The only property that needs to be checked then is the continuity of the expression within brackets in \eqref{eq:Bellman}.
	The first summand $\theta(p_t) - c$ is obviously continuous in $p_t$. The second summand can be rewritten as $V^*(\bar{p}) \cdot \left(\Phi(r(l,p_t) | p_t) - \Phi(l | p_t)\right)$. Lemma \ref{lem:eps} implies $r(l,p_t)$ is a continuous function. $\Phi(b|p_t)$ is a continuous function of its arguments, and, therefore, the summand is also continuous, as it is a composition of continuous functions. Finally, let $(l^n, p_t^n)$ be a sequence that converges to $(l,p_t)$, and let $r^n := r(l^n, p_t^n)$. Then because $f(b) \leqslant \frac{H-c}{1-\beta}$, by the Dominated Convergence Theorem (see \citet*{EG}, Theorem 1.19), we have that as $n \rightarrow +\infty$,
	\begin{equation*}
		\int\limits_r^1 f(b) \cdot \phi(b|p_t^n) \; db \rightarrow \int\limits_r^1 f(b) \cdot \phi(b|p_t) \; db.
	\end{equation*}
	Finally, because $\Phi(b|p_t)$ is continuous, and $r^n \rightarrow r$, $p_t^n \rightarrow p_t$,
	\begin{equation*}
		\int\limits_{r^n}^r f(b) \cdot \phi(b|p_t^n) \; db \leqslant \frac{H-c}{1-\beta} \cdot \left(\Phi(r|p_t^n) - \Phi(r^n|p_t^n)\right) \rightarrow 0,
	\end{equation*}
	which concludes the proof.

	\subparagraph{Step \thepfstep.2}
	\textbf{$\mathcal{C^I}[\bar{p}, 1]$ is a complete metric space.}
	
	The space of continuous functions defined on $[\bar{p},1]$ is complete. Therefore, any fundamental sequence $\{f_n\}_{n=1}^{+\infty}$ with $f_n \in \mathcal{C^I}[\bar{p}, 1]$ converges to a continuous function $f$. The remaining properties of the limiting function are straightforward. If $f_n(\bar{p}) \geqslant 0$ and $f(1) = \frac{H-c}{1-\beta}$ for all $n$, then taking the limit in $n$ we get $f(\bar{p}) \geqslant 0$ and $f(1) = \frac{H-c}{1-\beta}$. It remains to show the monotonicity of the limiting function. Fix any $x_1 < x_2$. If $f_n(x_1) \leqslant f_n(x_2)$ for all $n$, then taking the limit on both sides in $n$ we get $f(x_1) \leqslant f(x_2)$.

	\subparagraph{Step \thepfstep.3}
	\textbf{$T_f(p_t)$ is a contraction on $\mathcal{C^I}[\bar{p}, 1]$.}
	
	Take any $f,g \in \mathcal{C^I}[\bar{p}, 1]$ and $p_t \in [\bar{p},1]$. Denote the respective maximizers of \eqref{eq:mapping} at $p_t$ as $l_f, r_f$ for $f$ and $l_g, r_g$ for $g$ (if the maximizers are not unique, pick any). Without loss, assume that $T_f(p_t) \geqslant T_g(p_t)$. Then
	\begin{align*}
		T_f(p_t) - T_g(p_t) &= \beta \cdot \int\limits_{l_f}^{r_f} f(\bar{p}) \; d\Phi(b | p_t) + \beta \cdot \int\limits_{r_f}^{1} f(b) \; d\Phi(b | p_t) - \beta \cdot \int\limits_{l_g}^{r_g} g(\bar{p}) \; d\Phi(b | p_t) - \beta \cdot \int\limits_{r_g}^{1} g(b) \; d\Phi(b | p_t) \\
		& \leqslant \beta \cdot \int\limits_{l_g}^{r_g} \left[f(\bar{p}) - g(\bar{p})\right] \; d\Phi(b | p_t) + \beta \cdot \int\limits_{r_g}^{1} \left[f(b)-g(b)\right] \; d\Phi(b | p_t) \\
		& \leqslant \beta \cdot \rho(f,g) \cdot  \int\limits_{l_g}^{1}  d\Phi(b|p_t) \leqslant \beta \cdot \rho(f,g).
	\end{align*}
	Therefore, $\rho(T_f, T_g) = \max\limits_{p_t \in [\bar{p}, 1]} | T_f(p_t) - T_g(p_t) | \leqslant \beta \cdot \rho (f,g)$, and thus $T_f(p_t)$ is a contraction on $\mathcal{C^I}[\bar{p}, 1]$.

	\medskip
	The sub-steps above imply that there exists a unique $V^*(p_t)$ that solves \eqref{eq:Bellman}. Therefore, the associated maximizers $l(p_t), r(p_t)$ determine the optimal pooling interval for the optimal communication strategies. Step 10 above implies that these bounds are then unique for any $p_t \in [\bar{p}, 1)$. This concludes the proof of the theorem. \qed

	\subsection{Proofs for Section \ref{sec:CT}}

	In this section, we prove the results for the infinite-horizon model with cheap talk. We use ``cheap talk equilibrium'' to refer to ``stationary MPE with Cheap Talk''.
	Before we proceed with the proofs, note that Lemma \ref{lem:representation} applies to cheap talk equilibria as well. The notions of $V(p_{t+1} \mid p_t,b_t)$, $V^H (p_t)$, $V^L (p_t)$ are well defined, although numerically they can be different from the respective value functions for the commitment case.

	\begin{lemma} \label{lem:VH_VL}
		Fix any cheap talk equilibrium.
		Suppose that for some $p_t \in [\bar{p},1)$, there exist $m', m'' \in \mathcal{E}(p_t)$ such that $p_{t+1}' := q(m' | p_t) < q(m'' | p_t) := p_{t+1}''$.
		Then one of the following must hold:
		\begin{enumerate}
			\item $V^H(p_{t+1}') = V^H(p_{t+1}'')$ and $V^L(p_{t+1}') = V^L(p_{t+1}'')$;
			\item $V^H(p_{t+1}') < V^H(p_{t+1}'')$ and $V^L(p_{t+1}') > V^L(p_{t+1}'')$.
		\end{enumerate}
	\end{lemma}
	\begin{proof}
		By contradiction, if $V_\theta (p_{t+1}') < V_\theta (p_{t+1}'')$ for both $\theta$, then for any $b_t$: $V(m' \mid p_t, b_t) < V(m'' \mid p_t, b_t)$, so sending message $m'$ can never be optimal. Hence $m' \notin \mathcal{M}(p_t)$ -- a contradiction. Analogously, it cannot be that $V_\theta (p_{t+1}') > V_\theta (p_{t+1}'')$ for both $\theta$. 
		To complete the proof, we are left to rule out the case $V^H(p_{t+1}') > V^H(p_{t+1}'')$ and $V^L(p_{t+1}') < V^L(p_{t+1}'')$. Again, assume, by way of contradiction, that this is the case. Since $m',m'' \in \mathcal{E}(p_t)$, there exist $b', b'' \in [0,1]$, for which sending $m'$ and $m''$ respectively is optimal: $V(m' \mid p_t, b') \geqslant V(m'' \mid p_t, b')$ and $V(m' \mid p_t, b'') \leqslant V(m'' \mid p_t, b'')$.
		From Lemma \ref{lem:representation} it is then immediate that there exists $\bar{b} \in (0,1)$ such that $V(m' \mid p_t, b) \leqslant V(m'' \mid p_t, b)$ for $b \in [0, \bar{b}]$ and $V(m' \mid p_t, b) \geqslant V(m'' \mid p_t, b)$ for $b \in [\bar{b}, 1]$. But then $m' > \bar{b} > m''$, which contradicts the assumption that $m' < m''$. This concludes the proof.
	\end{proof}
	
	\begin{corollary} \label{cor:monotonicity}
		In any cheap talk equilibrium, for any $p_t \in [\bar{p},1)$:
		\begin{enumerate}
			\item $V^H(p_{t+1})$ is a (weakly) increasing function of $p_{t+1}$ on $[\bar{p},1)$,
			\item $V^L(p_{t+1})$ is a (weakly) decreasing function of $p_{t+1}$ on $[\bar{p},1)$.
		\end{enumerate}
	\end{corollary}
	\begin{proof}
		Follows directly from Lemma \ref{lem:VH_VL}.
	\end{proof}
	
	We will be using the following shorthand notation for the consumers' optimal continuation value in private state $(p_t, b_t)$:
	\begin{equation*}
		V(p_t,b_t) := \max \limits_{p_{t+1} \in \mathcal{P}(p_t)} V(p_{t+1} \mid p_t,b_t).
	\end{equation*}

	\begin{lemma} \label{lem:threshold1}
		In any cheap talk equilibrium, in any public state $p_t \in [\bar{p},1)$ there exists $\bar{b} (p_t) \in [0,1]$ such that
		\begin{enumerate}
			\item For all $b < \bar{b}(p_t)$,  $\mu(m \mid p_t, b) > 0$ only if $m \in \mathcal{S}(p_t)$;
			\item For all $b > \bar{b}(p_t)$, $\mu(m \mid p_t, b) > 0$ only if $m \in \mathcal{E}(p_t)$.
		\end{enumerate}
	\end{lemma}
	\begin{proof}
		The case when $\mathcal{S}(p_t) = \emptyset$ is equivalent to assuming $\bar{b}(p_t) = 0$. Therefore, we hereinafter assume that $\mathcal{S}(p_t) \ne \emptyset$. Assume the contrary: there exist $0 \leqslant b' < b'' \leqslant 1$ and $q' \in \mathcal{E}(p_t)$, $q'' \in \mathcal{S}(p_t)$ such that $\mu(q' | p_t, b') > 0$ and $\mu(q'' | p_t, b'') > 0$. Then $V(q'' \mid p_t, b'') = V(p_t, b'') = 0$ and $V(p_t, b') = V(q' \mid p_t, b') \geqslant V(q'' \mid p_t, b') = 0$. At the same time, $b' < b''$ and representation \eqref{eq:representation} from Lemma \ref{lem:representation} implies $V(q' \mid p_t, b'') > V(q' \mid p_t, b')$,  meaning
		\begin{equation*}
			0 = V(p_t, b'') \geqslant V(q' \mid p_t, b'') > V(q' \mid p_t, b') = V(p_t, b') \geqslant 0,
		\end{equation*}
		which is a contradiction. This argument proves the lemma.
	\end{proof}

	\begin{lemma} \label{lem:prelim}
		In any cheap talk equilibrium, in any public state $p_t \in [\bar{p},1)$,
		\begin{enumerate}
			\item $V(p_{t+1} | p_t, \bar{p}) \geqslant 0$ for any $p_{t+1} \in \mathcal{E}(p_t)$.
		  \item If there exists $p_{t+1} \in \mathcal{E}(p_t)$ s.t. $\mathcal{S}(p_{t+1}) \ne \emptyset$, then $V(p_t,\bar{p}) > 0$,
		\end{enumerate}
	\end{lemma}
	\begin{proof}
		We next prove the first part. Representation \eqref{eq:representation} applies and therefore
		\begin{equation} \label{eq:lem13}
			V(p_{t+1} \mid p_t, \bar{p}) = \beta \cdot \Big[\bar{p} \cdot V^H(p_{t+1}) + (1-\bar{p}) \cdot V^L(p_{t+1})\Big].
		\end{equation}
		For any $p_t$ denote as $p^i_{t+1}$, $i=1,2, \ldots$ the set of all posteriors that constitute $\mathcal{P}(p_{t})$, and by $w^\theta(p^i_{t+1} | p_t)$ the measure of the region where $p^i_{t+1}$ is induced given public belief $p_t$ and state $\theta$ according to $\Phi_\theta(b|p_t)$. Without loss, assume that if $\mathcal{S}(p_t) \ne \emptyset$ then $\mathcal{S}(p_t) = \{p^1_{t+1}\}$ and that posteriors are labeled in the increasing order: $p^i_{t+1} < p^{i+1}_{t+1}$ for all $i$. Depending on whether $\mathcal{S}(p_t) = \emptyset$ or not \eqref{eq:lem13} expands differently. First, consider the case when $\mathcal{S}(p_t) \ne \emptyset$. Then \eqref{eq:lem13} reduces to
		\begin{multline} \label{eq:posterior}
			V(p_{t+1} \mid p_t, \bar{p}) = \beta \cdot \Big[ \bar{p} \cdot (1-w^H(p^1_{t+2} | p_{t+1})) \cdot (H-c) + (1-\bar{p}) \cdot (1-w^L(p^1_{t+2} | p_{t+1})) \cdot (L-c) + \\ \bar{p} \cdot \beta \cdot \sum\limits_{i \geqslant 2} w^H(p^i_{t+2} | p_{t+1}) \cdot V^H(p^i_{t+2}) + (1-\bar{p}) \cdot \beta \cdot \sum\limits_{i \geqslant 2} w^L(p^i_{t+2} | p_{t+1}) \cdot V^L(p^i_{t+2})\Big].
		\end{multline}
		Note that $\Phi_H(b | p_t)$ first-order stochastically dominates $\Phi_L(b | p_t)$. Therefore, $0 < w^H(p^1_{t+2} | p_{t+1}) < w^L(p^1_{t+2} | p_{t+1}) < 1$ and thus the sum of the first two terms in \eqref{eq:posterior} is strictly positive. Moreover, by Corollary \ref{cor:monotonicity} we then have
		\begin{equation*}
			\sum\limits_{i \geqslant 2} w^H(p^i_{t+2} | p_{t+1}) \cdot V^H(p^i_{t+2}) \geqslant \sum\limits_{i \geqslant 2} w^L(p^i_{t+2} | p_{t+1}) \cdot V^H(p^i_{t+2}).
		\end{equation*}
		$V(p_{t+1} \mid p_t, \bar{p})$ then satisfies
		\begin{equation*}
			V(p_{t+1} \mid p_t, \bar{p}) > \beta^2 \cdot \Big[ \sum\limits_{i \geqslant 2} w^L(p^i_{t+2} | p_{t+1}) \cdot \left( \bar{p} \cdot V^H(p^i_{t+2}) + (1-\bar{p}) \cdot V^L(p^i_{t+2})\right)\Big].
		\end{equation*}
		By representation \eqref{eq:representation} the term in brackets reduces to $V(p^i_{t+2} | p_{t+1}, \bar{p})$. Therefore
		\begin{equation} \label{eq:main1}
			V(p_{t+1} \mid p_t, \bar{p}) > \beta \cdot \Big[ \sum\limits_{i \geqslant 2} w^L(p^i_{t+2} | p_{t+1}) \cdot V(p^i_{t+2} | p_{t+1}, \bar{p})\Big].
		\end{equation}
		The case when $\mathcal{S}(p_t) = \emptyset$ is similar. Equation \eqref{eq:lem13} transforms to
		\begin{align*}
			V(p_{t+1} \mid p_t, \bar{p}) = \beta^2 \cdot \Big[ \bar{p} \cdot \sum\limits_{i \geqslant 1} w^H(p^i_{t+2} | p_{t+1}) \cdot V^H(p^i_{t+2}) + (1-\bar{p}) \cdot \sum\limits_{i \geqslant 1} w^L(p^i_{t+2} | p_{t+1}) \cdot V^L(p^i_{t+2})\Big].
		\end{align*}
		Proceeding through the same steps we get
		\begin{equation} \label{eq:main2}
			V(p_{t+1} \mid p_t, \bar{p}) \geqslant \beta \cdot \Big[ \sum\limits_{i \geqslant 1} w^L(p^i_{t+2} | p_{t+1}) \cdot V(p^i_{t+2} | p_{t+1}, \bar{p})\Big].
		\end{equation}
		We can iterate inequalities for \eqref{eq:main1} and \eqref{eq:main2} further. $V(p_{t+1} \mid p_t, b_t)$ is bounded in its absolute value due to \eqref{eq:representation} by $M:= \max \{\frac{H-c}{1-\beta}, \frac{c-L}{1-\beta}\}$ and the sum of all weights in the RHS in \eqref{eq:main1} and \eqref{eq:main2} does not exceed $1$. Therefore for sufficiently many iteration steps the absolute value of the RHS in \eqref{eq:main1} and \eqref{eq:main2} can be made smaller than any positive number. This clearly implies that $V(p_{t+1} \mid p_t, \bar{p}) \geqslant 0$.

		The second part of the lemma follows from the first part and representation \eqref{eq:main1} because $V(p^i_{t+2} | p_{t+1}, \bar{p}) \geqslant 0$.
	\end{proof}

	\begin{lemma} \label{lem:threshold2}
		Threshold $\bar{b}(p_t)$ defined in Lemma \ref{lem:threshold1} satisfies $\bar{b} (p_t) \in [0, \bar{p}]$. Further, $\bar{b}(p_t) = \bar{p}$ if and only if the experimentation never stops after any $m \in \mathcal{E}(p_t)$.\footnote{It means that $\mathcal{S}(p_\tau) = \emptyset$ for all $p_\tau, \tau > t$ which can be on path of play originating from $p_t$.}
	\end{lemma}
	\begin{proof}		
		If experimentation never stops after any $p_{t+1} \in \mathcal{E}(p_t)$, $C_t$ decides whether all future consumers buy the product or avoid it. The discounted expected utilities in the two cases from her point of view is then $\frac{1}{1-\beta} \cdot \left(\theta(b_t) - c\right)$ and zero, respectively. Therefore, $C_t$ sends $p_{t+1} \in \mathcal{E}(p_t)$ if and only if $\theta(b_t) \geqslant c$ -- equivalently, $b_t \geqslant \bar{p}$, -- and stops experimentation otherwise.
		
		It is left to show the converse -- that if experimentation can be stopped by some future consumer then $\bar{b}(p_t) < \bar{p}$. 
		Let $\tau(p_t) := \min \{ \tau \geqslant 1 \mid \exists \; p_{t+1} \in \mathcal{E} (p_t), p_{t+2} \in \mathcal{E} (p_{t+1}), \ldots, p_{t+\tau} \in \mathcal{E} (p_{t+\tau-1}): \mathcal{S}(p_{t+\tau}) \ne \emptyset \}$ denote the minimal number of periods that should pass before one of the following consumers, $C_{\tau(p_t)}$, has an opportunity to stop experimentation. We show that if $\tau(p_t) < \infty$, then $V(p_t, \bar{p}) > 0$, which is equivalent to $\bar{b}(p_t) < \bar{p}$.
		The previous lemma proves the statement for $\tau(p_t) = 1$. We will proceed by induction on $\tau$.
		
		Suppose that induction statement ``if $\tau(p_t) = \tau - 1$ then $V(p_t, \bar{p}) > 0$'' is true for any $p_t$ for some $\tau \geqslant 2$. To see that it then also holds for $\tau(p_t) = \tau$, consider the following argument. Let $q$ denote the public posterior such that $q \in \mathcal{E}(p_t)$ and $\tau(q) = \tau - 1$ (in case of multiple such posteriors select arbitrarily). Then because $\mathcal{S}(q) = \emptyset$ representation \eqref{eq:main2} applies, and therefore
		\begin{equation*}
		   V(p_t, \bar{p}) \geqslant V(q \mid p_t, \bar{p}) \geqslant \beta \cdot w^L(p^1_{t+2} | q) \cdot V(p^1_{t+2} | q, \bar{p}).
		\end{equation*}
		Finally, notice that $p^1_{t+2}$ is the minimal posterior induced in public state $q$ and $\mathcal{S}(q) = \emptyset$. $p^1_{t+2}$ is induced for all $b_{t+1} \leqslant \bar{p}$ and is optimal for $b_{t+1} = \bar{p}$ in particular. Therefore $w^L(p^1_{t+2} | q) > 0$ and $V(p^1_{t+2} | q, \bar{p}) = V(q, \bar{p}) > 0$ by the induction hypothesis.
	\end{proof}
	
	\paragraph{Proof of Proposition \ref{prop:NPR_CT}.} 
	Suppose the contrary: not every $m \in \mathcal{M}(p_t)$ starts a cascade. Then by Lemma \ref{lem:threshold2} we have $V(p_t, \bar{p}) > 0$. Therefore $m \in \mathcal{M}(p_t)$, which is sent for $b_t = \bar{p}$, is sent for some $b_t < \bar{p}$ as well. Belief consistency then requires that it is sent for some $b_t > \bar{p}$ which finishes the argument.
	\qed	
	
	\paragraph{Proof of Theorem \ref{thm:dec}.} 
	Lemma \ref{lem:threshold2} implies $\bar{b}(p_t) < \bar{p}$. Therefore, we can take $l(p_t) = \frac{\bar{b}(p_t) + \bar{p}}{2} \in (\bar{b}(p_t), \bar{p})$. By belief consistency there exist $r(p_t) > \bar{p}$ and message $m \in \mathcal{M}(p_t)$ such that $m$ is sent for all $b_t \in [l(p_t), r(p_t)]$.
	\qed

\end{document}